\documentclass[11pt]{article}
\pdfoutput=1

\usepackage{booktabs} % For formal tables
\usepackage{natbib,xfrac}
\usepackage{algorithm}
\usepackage[noend]{algorithmic}
\usepackage{amsmath,amsthm,amsfonts}
\usepackage{amsmath}
\usepackage{color}
\usepackage{mathrsfs}
\usepackage{verbatim}
\usepackage{enumitem}
\usepackage{bm}
\usepackage{hyperref}
\usepackage{xr}
\usepackage{tikz}
\usepackage{subcaption}
\usepackage{float}
\usepackage{listings}
\usepackage{array}
\usepackage{makecell}
\usepackage{mathtools}
\usepackage[suppress]{color-edits}
\addauthor{nh}{magenta}
\addauthor{aw}{blue}
\addauthor{ll}{orange}

\usepackage{fancyhdr, lastpage}
\usepackage[vmargin=1.00in,hmargin=1.00in,centering,letterpaper]{geometry}

\newcommand{\investmentlevel}{qualification rate}
\newcommand{\drule}{assessment rule}
\newcommand{\drules}{assessment rules}
\newcommand{\investmentlevels}{qualification rates}
\newcommand{\optinvest}{optimal qualification rate~}
\newcommand{\ptp}{p_{\mathrm{TP}}}
\newcommand{\cfp}{c_{\mathrm{FP}}}

\newcommand{\fpr}{\mathrm{FPR}}
\newcommand{\tpr}{\mathrm{TPR}}

\newcommand{\otheta}{\theta^{\mathrm{opt}}}

\newcommand{\dec}{^\mathrm{dec}}
\newcommand{\sub}{^\mathrm{sub}}

\newcommand{\argmax}{\mathop{\rm argmax}}
\newcommand{\Ind}[1]{\mathbf{1}\{#1\}}

\newcommand{\norm}[1]{\left\|{#1}\right\|}

\newcommand{\trans}{^{\top}}

\renewcommand{\Pr}{\mathbb{P}}

\newcommand{\R}{\mathbb{R}}

\newcommand{\A}{\mathcal{A}}

\newcommand{\calX}{\mathcal{X}}
\newcommand{\calY}{\mathcal{Y}}

\newcommand{\calA}{\mathcal{A}}

\newtheorem{theorem}{Theorem}[section]

\newtheorem{proposition}[theorem]{Proposition}
\newtheorem{corollary}[theorem]{Corollary}

\newtheorem{lemma}[theorem]{Lemma}

\newtheorem{assumption}{Assumption}

\newcommand{\brpi}{\pi^{br}}
\newcommand{\brtheta}{\theta^{br}}
\renewcommand{\hm}{h_{\mathrm{mid}}}

\definecolor{DarkBlue}{rgb}{0.1,0.1,0.5}
\hypersetup{
	colorlinks=true,       % false: boxed links; true: colored links
	linkcolor=DarkBlue,          % color of internal links
	citecolor=DarkBlue,        % color of links to bibliography
	filecolor=DarkBlue,      % color of file links
	urlcolor=DarkBlue,          % color of external links
	pdftitle={},
	pdfauthor={},
}

\begin{document}

\title{The Disparate Equilibria of Algorithmic Decision Making when Individuals Invest Rationally}
%\author{}

\author{
	Lydia T. Liu\thanks{Correspondence to: lydiatliu@cs.berkeley.edu }\\
	University of California, Berkeley 
	\and Ashia Wilson\\
	Microsoft Research 
	\and Nika Haghtalab \\
	Cornell University 
	\and Adam Tauman Kalai \\
	Microsoft Research 
\and  Christian Borgs \\ 
Microsoft Research 
\and  Jennifer Chayes\\
Microsoft Research }
%\author{
%Lydia T. Liu\thanks{Department of Electrical Engineering and Computer Sciences, University of California, Berkeley}
%\and Ashia Wilson\thanks{Microsoft Research}
%\and Nika Haghtalab\thanks{Cornell University}
%\and Adam Kalai\footnotemark[2]
%\and Christian Borgs\footnotemark[2]
%\and Jennifer Chayes\footnotemark[2]}

\date{}

\maketitle

\begin{abstract}
	% !TeX root = main.tex 

The long-term impact of algorithmic decision making is shaped by the dynamics between the deployed decision rule and individuals' response. Focusing on settings where each individual desires a positive classification---including many important applications such as hiring and school admissions, we study a dynamic learning setting where individuals invest in a positive outcome based on their group's expected gain and the decision rule is updated to maximize institutional benefit. By characterizing the equilibria of these dynamics, we show that  natural challenges to desirable long-term outcomes arise due to heterogeneity across groups and the lack of realizability. We consider two interventions, decoupling the decision rule by group and subsidizing the cost of investment. We show that decoupling achieves optimal outcomes in the realizable case but has discrepant effects that \lledit{may} depend on the initial conditions otherwise. \lledit{In contrast, subsidizing the cost of investment is shown to create better equilibria for the disadvantaged group even in the absence of realizability.}%In contrast, subsidizing the cost of investment creates better equilibria for the disadvantaged group regardless of realizability. \llcomment{Is this statement too strong? e.g. there can be a weird case where subsidies results in multiple equilibria.}
\end{abstract}

% !TeX root = main.tex 

\section{Introduction}
%\nhcomment{I'm not completely sold on the first paragraph of the intro,  although I like the general idea it's trying to get across. My suggestion is to focus on paragraph 2-end first and then we can see what would be the most crisp way of doing the first paragraph.}

%Increasingly more, machine learning is being used to design policies that affect the entity it once learned about.

Automated decision-making systems that rely on machine learning are increasingly used for
%Machine learning algorithms are increasingly applied to 
high-stakes applications, yet their long-term consequences have been controversial and poorly understood.
On one hand, deployed decision making models are updated periodically to assure high performance on the target distribution.
On the other hand, deployed models can reshape the underlying populations thus biasing how the model is updated in the future.
This complex interplay between algorithmic decisions, individual-level responses, and exogeneous societal forces can lead to pernicious long term effects that reinforce or even exacerbate existing social injustices~
%the deployed decision-making model can reshape the underlying population and thus bias the future data that is used to update the model, to pernicious and often unexpected effects. 
%complex interplay between algorithmic decisions, individual-level responses, and exogeneous societal forces, can reinforce or even exacerbate existing social injustices 
\citep{crawford17trouble,whittaker2018ainow}. %, and in particular the marginalization of disadvantaged minorities
Harmful feedback loops have been observed in automated decision making in several contexts including recommendation systems~\citep{Pariser,Twitter,Chaney:2018confounding}, predictive policing~\citep{Ensign18}, admission decisions~\citep{Lowry657,Barocas16}, and credit markets~\citep{Fuster,Abhay19}. 
These examples underscore the need to better understand the dynamics of algorithmic decision making, in order to align decisions made about people with desirable long-term societal outcomes.

Automated decision-making algorithms rely on observable features to predict some variable of interest. In a setting such as hiring, \lledit{decision making models \emph{assess}} features such as scores on standardized tests, resume, and recommendation letters, to identify \nhedit{individuals that are ideally \emph{qualified} for the job.}
However, equally qualified people from different demographic groups tend to have different features, due to implicit societal biases (e.g., letter writers describe competent men and women differently), gaps in resources (e.g., affluent students can afford different extra-curricular activities) and even distinct tendencies in self-description (e.g., gender can be inferred from biographies \citep{De-Arteaga2019bias}). 
\lledit{Therefore, a model's ability }to identify qualified individuals can widely vary across different groups.

\lledit{The deployed model's ability to identify qualified members of a group affects an individual's incentive to \emph{invest} in their qualification. 
This is because} one's decision to acquire qualification---not observed directly by the algorithm---comes at a cost. 
\nhedit{Moreover, individuals that are identified by the model as qualified (whether or not they are truly qualified) receive a reward.}
Consequently, people invest in acquiring qualifications only when their expected reward from the \lledit{assessment model} beats the investment cost. 

Rational individuals are aware that upon investing they would develop features that are similar to those of qualified individuals in their group, so they gauge their own expected reward from investing by the observed rewards of their group.\footnote{Strong group identification effects can also be seen in empirical studies \citep{Hoxby}. } If qualified people from one group are not duly identified and rewarded, fewer people from that group are incentivized to invest \lledit{in qualifications} in the future. This reduces the overall fraction of qualified people in that group, or the \emph{qualification rate}. As the \nhedit{assessment} model is updated to maximize overall institutional profit on the new population distribution, it may perform even more poorly on qualified individuals from a group with relatively low \investmentlevel, further reducing the group's incentive to invest. %Similar effects have been studied under the topic of statistical discrimination in economics \citep{Arrow73theory,coate93will, Arrow98economics}.

To understand and mitigate the challenges to long-term welfare and fairness posed by such dynamics, we propose a formal model of sequential learning and decision-making where at each round a new batch of individuals {rationally decide whether to invest in acquiring qualification and the institution updates its assessment rule (a classifier) for assessing and thus rewarding individuals.}
We study the long-term behavior of these dynamics by characterizing their equilibria and comparing these equilibria based on several metrics of social desirability.
%make rational decisions to invest and the institution's decision boundary is updated on the current population distribution (Section~\ref{sec:model}). 
Our model can be seen as an extension of \citet{coate93will}'s widely cited work \lledit{to explicitly address heterogeneity in observed features across groups.
%\nhedit{to group-specific feature distributions.} 
While \citet{coate93will}'s model focuses on a single-dimensional feature space, i.e., scores, and assessment rules that act as thresholds on the score, our model considers general, possibly high-dimensional, feature spaces and arbitrary assessment rules, which are typical in high-stakes domains such as hiring, admissions, etc.} %We also explicitly address heterogeneity in observed features across groups, in a departure from their work.

We find that {two major obstacles to obtaining desirable long-term outcomes are}
heterogeneity across groups and lack of realizability within a group.
%leads to undesirable tradeoffs between metrics at equilibria, resulting in long term outcomes that disadvantage one or more groups. 
Realizability, {which is the existence of a (near) perfect way to assess qualifications of individuals from visible features, leads to equilibria that are (near) optimal on several metrics, such as the resulting qualification rates,} their uniformity across groups, and the institution's utility.
{We study (near) realizability and \lledit{the} lack thereof in Sections~\ref{sec:realizable} and~\ref{sec:non-realizable} respectively.}
%decision boundary that predicts qualification (almost) perfectly given the available features---precludes convergence to undesirable equilibria (Section~\ref{sec:realizable} and~\ref{sec:non-realizable}).
Heterogeneity across groups, {i.e., lack of a single assessment rule that perfectly assesses individuals from all groups,} necessitates tradeoffs in the {quality of equilibria across different groups.
We study heterogeneity, as well as interventions for mitigating its negative effects, in Section~\ref{sec:groups}.}
%\investmentlevel~for different groups as no one decision boundary is optimal for all groups (Section~\ref{sec:groups}).
In Section~\ref{sec:experiments}, we empirically study a more challenging setting where the groups are heterogeneous as well as highly non-realizable. For this we \lledit{perform
% arethe prediction problem is highly non-realizable and groups have different feature distributions, 
simulations with a FICO credit score dataset~\citep{fed07} that has been widely used for illustration} in the algorithmic fairness literature.

%\vspace{1mm}
\paragraph{Interventions}
%\emph{Interventions.}
To mitigate the aforementioned tradeoffs, we consider two common interventions: decoupling the decision policy by group and subsidizing the cost  of investment, {especially when the cost distribution inherently differs by group.
%\nhcomment{Overall a good paragraph, but massage it so it would read nicer. You may want to use itemize or enumerate  or list so the contributions become more visually clear.}
Our model of dynamics sheds a different light on these interventions, complementary to previous work.
We show that decoupling~\citep{dwork18a}---using {group-specific assessment rules}---achieves optimal outcomes when the problem is realizable within each group, but can significantly hurt certain groups when the problem is non-realizable \lledit{and there exist multiple equilibria after decoupling. In particular, decoupling can hurt a group with low initial \investmentlevel~if the utility-maximizing assessment rule for a single group is more disincentivizing to individuals than a joint {assessment rule}, thereby reinforcing the status quo and preventing the group from reaching an equilibrium with higher \investmentlevel.} % In particular, decoupling can hurt a group with low initial \investmentlevel~by reinforcing the status quo and preventing the group from reaching a better equilibrium. This is because {a utility-maximizing assessment rule} for a single group with low \investmentlevel~could be significantly more disincentivizing to individuals than a joint {assessment rule}, in the non-realizable setting. %Our findings suggest one way in which a well-meaning algorithmic intervention might in a sense dampen upward social mobility, especially in settings where accurate prediction is inherently impossible.

%\nhcomment{Perhaps bring the sentence ``Our model of dynamics sheds a different light on each of these interventions, complementary to findings in previous work'' here so it signals that we are going to compare and contrasts our results with a bunch of things now.}

We also study subsidizing individuals' investment cost (e.g. subsidizing tuition for a top high school), especially when the cost distribution is varied across different groups. We {find} that {these subsidies} {increase the qualification rate of} the disadvantaged group at equilibrium, regardless of realizability.
We note that our subsidies, which affect the qualification of individuals directly, are different than those studied under strategic manipulation~\citep{Hu2019disparate} that involve subsidizing individual's cost to manipulate their features \emph{without changing the underlying qualification} (e.g. subsidizing SAT exam preparation \lledit{without changing the student's qualification for college}) and could have adverse effects on disadvantaged groups.
%different from subsidizing individuals to manipulate their features without changing their underlying qualification (e.g. subsidizing SAT preparation) as formulated in the related but distinct setting of strategic manipulation \citep{Hu2019disparate}.
%individual to invest directly in their qualification (e.g. subsidizing tuition for top high schools) on the equilibrium \investmentlevel~for disadvantaged groups. 
%This is different from subsidizing individuals to manipulate their features without changing their underlying qualification (e.g. subsidizing SAT preparation) as formulated in the related but distinct setting of strategic manipulation \citep{Hu2019disparate}. Unlike in that setting, we found that in our setting subsidizing the cost of investment has a positive impact on the disadvantaged group at equilibrium, regardless of realizability. 
Instead, our theoretical findings resonates with extensive empirical work in economics on the effectiveness of subsidizing opportunities for a disadvantaged group to \lledit{directly improve their outcomes, such as moving to better neighborhoods to access} better educational and environmental resources~\citep{Chetty}. % which 
%\todo{Ashia, please add some details. Also missing citation. Thanks!}

%\vspace{1mm}
\paragraph{Related work}
%\emph{Related work.}
Our work is related to a rich body of work on algorithmic fairness in dynamic settings \citep{liu18delayed,Hu2018shortterm,hashimoto18demographics,zhang2019longterm,Mouzannar2019socialequality}, strategic classification \citep{Hu2019disparate,Milli2019social,kleinberg18investeffort}, as well as statistical discrimination \lledit{in economics} \citep{Arrow73theory,coate93will,Arrow98economics} We present a detailed discussion of the similarities and differences in Section~\ref{sec:related}.

%
%\paragraph{Main contribution of paper}
%\begin{itemize}
%	\item {\bf Proposing model:} We propose a model under which we consider rational investment by individuals. This model allows us to think about the properties of equilibria 
%	\item {\bf Understanding challenges:} We are showing that some natural challenges arise from real data and properties of machine learning models 
%	\begin{itemize}
%		\item When model is unrealizable, dynamics may tend towards equilibria where investment is fundamentally unbalanced amongst groups, even when \dots  cost of investment are the same or when \dots %the groups are can be symmetric/exchangeable? Percy�s wording. 
%		\item In general, the stable equilibria favors one group. 
%		\item Also, when non-realizability within group there exists multiple equlibria within group.
%	\end{itemize}
%	\item {\bf Proposing solutions:} Unlike the strategic classification, subsidizing cost of investment seems to have significant impact both in theory and in practice
%	\begin{itemize} 
%		\item Impact of providing subsidies
%		\item  Impact of decoupling 
%	\end{itemize}
%\end{itemize}

% !TeX root = main.tex 

\section{A Dynamic Model of Algorithmic Decision Making }\label{sec:model}

\begin{figure}[t]
\setlength{\belowcaptionskip}{-10pt}
	\centering
	\begin{tikzpicture}[->,shorten >=1pt,auto,node distance=1.5cm,
	thick,main node/.style={circle,draw,font=\Large},sub node/.style={font=\Large}]
	
	\node[main node] (Y) {$Y$};
	\node[main node] (X) [below of=Y] {$X$};
	\node[sub node] (D) [left of=Y] {$y$};
	\node[main node] (A) [above right of=X] {$A$};
	
	\path[every node/.style={font=\small}]
	(Y) edge node [right] {} (X)
	(D) edge node [right] {} (Y)
	(A) edge node [right] {} (X);
	\end{tikzpicture}

	\caption{Causal graph for the individual investment model. The individual intervenes on the node for qualification, $Y$---this corresponds to do$(Y=y)$)---which then affects the distribution of their features $X$, depending on the group~$A$.
	}
	\label{fig:graph-ourmodel}
\end{figure}
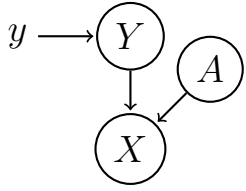

%\todo{The figure caption talks about $c$. Is that $C$? Also doesn't $y$ also depend on $w$?Perhaps we want to be a little less exact here to make the caption nicer.}

In this section we introduce a model of automated decision making with feedback. We start by introducing the notation used throughout the paper and then describe the details of the interactions between individuals and an institution, and the resulting dynamical system.

%the difference between value and random variable and not get confused by but it's possible that somebody that's a bit less savey with probability might find the switch between values and random variables confusing.} 
\subsection{Notation}
We consider an instance space $\calX$, where $X\in \calX$ denotes the features of an individual that are observable by the institution. We also consider a label space $\calY = \{0,1\}$ where $Y=1$ indicates that an individual has the qualifications desired by the institution and $Y=0$ otherwise. We denote the set of all protected/group attributes by $\calA$ where $A\in \A$ denotes an individual's protected attribute. 
%A qualified individual, that is an individual with $Y=1$, from group $a\in \calA$ has features that are distributed according to $\Pr(X=x \mid Y=1, A=a)$ (see Figure~\ref{fig:graph-ourmodel}). Similarly, an unqualified individual from group $a\in \calA$ has features that are distributed according to $\Pr(X=x\mid Y=0, A=a)$.
We denote the group proportions by $n_a \coloneqq \Pr(A=a)$ for all $a\in \calA$. 
Furthermore, we denote the \investmentlevel~in group $a\in \calA$ by
$\pi_a \coloneqq \Pr(Y=1 \mid A = a)$.  
%
%After making this decision a qualified individual from group $a$ receives features $x\in \calX$ with probability $\Pr(X=x \mid Y=1, A=a)$ and an unqualified individual from group $a$ receives features $x\in \calX$ with probability $\Pr(X=x \mid Y=0, A=a)$.
%
An individual from group $A=a$ who has acquired label $Y=y$ (to become qualified or not)\footnote{This can be seen as the individual performing a do-intervention on $Y$ \citep[see e.g.,][]{Pearl:2009:CMR:1642718}. Thus we may write $do(Y=1)$ for making the decision to acquire qualifications. Our model (Figure~\ref{fig:graph-ourmodel}) assumes that $Y$ is not the child of any node, so we in fact have $\Pr(\cdot \mid do(Y=y)) = \Pr(\cdot \mid Y=y)$. Therefore we drop the $do$-operator whenever we condition on $Y$.} receives features $X$ distributed according to $\Pr(X=x \mid Y=y, A=a)$. This is illustrated in Figure~\ref{fig:graph-ourmodel}.

We also consider a set of parameters $\Theta$ that are used for assessing  qualifications. We use $\hat{Y}_\theta \in \calY$ parameterized by $\theta \in \Theta$ to denote the \emph{assessed qualification} of an individual. We assume that $\hat{Y}_\theta$ only depends on the features $X$\lledit{, which may or may not contain $A$ or its proxies.} \lledit{In later sections, }we will also discuss interventions that allow us to use $\hat{Y}_\theta$ that explicitly depends on group membership $A$.
We respectively define the \emph{true positive rate} and \emph{false positive rate} of $\theta\in \Theta$  on group $a\in \calA$ by 
\begin{align*}
	\tpr_a(\theta) &= \Pr(\hat{Y}_\theta = 1 \mid  Y=1, A=a), \text{ and } \\
	\fpr_a(\theta) &= \Pr(\hat{Y}_\theta = 1 \mid Y=0 , A=a).
	\end{align*}

\subsection{Model Description}

\paragraph{Individual's Rational Response}
We consider a setting where an individual decides whether to acquire qualifications, that is, to invest in obtaining label $Y=1$, prior to observing their feature $X$.
The decision to acquire qualification depends on the qualification assessment rule $\theta\in \Theta$ currently implemented \lledit{by the institution.} %in the environment.
We will characterize the groups' qualification rates as the \emph{best-response} to $\theta$ by function $\brpi(\theta) = (\brpi_a(\theta))_{a\in \calA}$.

To get label $Y=1$ an individual has to pay a cost $C > 0$. In any group, $C$ is distributed randomly according to the cumulative distribution function (CDF), $G(\cdot)$.\footnote{For the rest of this work, unless otherwise stated, we assume that the distribution of costs, $G$, is the same for every group. In Section~\ref{sec:subsidies} and \ref{sec:experiments}, we will consider the implications of having different cost distributions by group.} After deciding whether to acquire qualifications, an individual gets features $X$ and is assessed by $\theta$. An individual (from any group and regardless of actual qualification) receives a payoff of $w > 0$ if they are assessed to be qualified and payoff of $0$ otherwise. 
Therefore, the expected utility an individual from group $a$ receives from acquiring qualification $Y=1$ is $w \Pr[\hat Y_\theta =1 | Y=1, A=a] - C = w \tpr_a(\theta) - C$ whereas the expected utility for not acquiring the qualification is $w \Pr[\hat Y_\theta =1 | Y=0, A=a] = w \fpr_a(\theta)$.
Given the qualification assessment parameter $\theta\in \Theta$, an individual from group $a$ acquires qualification if and only if the benefit outweighs the costs, that is
 \begin{equation}\label{eq:individ_dec}
w (\tpr_a(\theta) - \fpr_a(\theta) ) > C.
\end{equation}

Then each group's qualification rate as a function of a qualification assessment parameter $\theta$ is
\begin{align*}
\brpi_a(\theta) \coloneqq \Pr(Y=1\mid A=a) &= \Pr( C < w (\text{TPR}_a(\theta) - \text{FPR}_a(\theta) ) ) \\
&= G(w (\text{TPR}_a(\theta) - \text{FPR}_a(\theta) )).
\end{align*}

\paragraph{Institution's Rational Response} 
We consider an institution that \lledit{has to choose} a qualification assessment parameter for accepting individuals \lledit{to maximize its utility}.
We assume that the institution gains $p_{\text{TP}} > 0$ for accepting a qualified individual and %hiring a skilled candidate and 
loses $c_{\text{FP}} > 0$ for accepting an unqualified individual. %hiring an unskilled candidate. 
Then the expected utility of the institution for applying parameter $\theta$ is
\begin{align*}
&\ptp  \Pr(\hat{Y}_\theta = 1, Y=1) - \cfp  \Pr(\hat{Y}_\theta = 1, Y=0)\\ =~& 
\ptp \sum_{a\in \calA} \tpr_a(\theta)   \pi_a   n_a -\sum_{a\in \calA} \cfp  \fpr_a(\theta)   (1-\pi_a)   n_a.
\end{align*}
This illustrates that \lledit{the utility maximizing parameter} is a function of $\pi = (\pi_a)_{a\in \calA}$, i.e., the rate of qualification in  each group. We denote this function by $\brtheta(\pi)$, defined as follows:
\begin{equation}\label{eq:instit_dec}
\brtheta(\pi) := \argmax_{\theta\in \Theta} ~ \ptp  \sum_{a\in \calA} \tpr_a(\theta)   \pi_a   n_a -\sum_{a\in \calA} \cfp  \fpr_a(\theta)   (1-\pi_a)   n_a.
\end{equation}
%The identity, 
%%Note that 
% \begin{align*}
%&\argmax_{\theta\in \Theta} ~p_{\text{TP}}\cdot \Pr(\hat{Y}_\theta = 1, Y=1) - c_{\text{FP}}\cdot \Pr(\hat{Y}_\theta = 1, Y=0) \\
%	=&\argmax_{\theta\in \Theta} ~p_{\text{TP}}\cdot \sum_{a\in \calA} \text{TPR}_a(\theta) \cdot \pi_a \cdot n_a -\sum_{a\in \calA} c_{\text{FP}}\cdot \text{FPR}_a(\theta) \cdot (1-\pi_a) \cdot n_a,
%\end{align*}
%illustrates that the utility maximizing set of parameter $\theta^\ast$ is a function of $\pi $, the current proportion of qualified vs. unqualified people in each group.
To ensure the above object (and the resulting dynamics) are well-defined, when multiple parameters~$\theta$ achieve the optimal utility we assume that $\brtheta(\pi)$ is uniquely defined using a fixed and well-defined tie-breaking function.

We note that throughout this paper we assume that the institution has exact knowledge of many quantities such as $\tpr_a(\theta)$, $\fpr_a(\theta)$, and $n_a$. In a nutshell, we assume that we have infinitely many samples from the underlying distributions. We discuss this further in Section~\ref{sec:discussion}, and leave the finite sample version of these results to future work.

Although we choose not to focus on game-theoretical aspects in this work, we note that our model can be thought of as a large game \citep{Kalai04largegames} or a game with a continuum of players \citep{Schmeidler1973}.

\paragraph{Dynamical System and Equilibria}
We are primarily interested in the evolution of qualification rate, $\pi$, over time. 
Given a current  rate of qualification $\pi$ the assessment parameter used by the institution in the next step is $\brtheta(\pi)$, which in turn leads to a qualification rate of 
$\brpi(\brtheta(\pi))$ in the next step. Therefore, we define a dynamical system for a given initial state $\pi(0)$ such that at time $t$ we are in state
$\pi(t)  = \Phi(\pi({t-1}))$, where $\Phi = \brpi \circ \brtheta$.

%$\pi$ evolves according to $\pi(t)  = \Phi(\pi({t-1}))$, where $\Phi = \pi \circ \theta^*$. %We assume that  $\theta^\ast$ is a uniquely determined so that $\Phi$ is well-defined. %, , $\Phi$ is not well-defined, and so we. To resolve this issue, we assume \lnote{resolve notation when $\theta^*$ is a set?}%$\dots$  \lnote{resolve notation when $\theta^*$ is a set?}
%\paragraph{Equilibria}
We say that the aforementioned dynamical system is at equilibrium if $\pi = \Phi(\pi)$. Equivalently, we are at an equilibrium if $\pi =  \lim_{n \rightarrow \infty} \Phi^n(\pi(0))$ is well-defined for some $\pi(0)$, where $\Phi^n$ is an $n$-fold composition of $\Phi$. We call such values of $\pi$ %(and their corresponding $\theta$)
\emph{equilibria}, or equivalently, \emph{fixed points} of $\Phi$.% That is when $\pi$ is at an equilibrium, it will stay at that value. 

%Interesting questions are:
%\begin{enumerate}
%	\item Are there multiple equilibria? How do we select among them? (e.g. depending on initial conditions)
%	\item Number of equilibria in multi-dimensional $X$ setting, and characterization of $\pi_a$ at these equilibria
%	\item View selection bias as a way of selecting among equilibria and show it can be hard to correct without large changes?
%	\item Does the choice of model class/parametrization have ramifications for the types of equilibria?
%	\item Does a bound on differences in feature distributions give a bound on the group-differences at equilibria?
%	\item Effect of collecting more features on tradeoffs at equilibria?
%	\item Non-realizability within groups (i.e. conditional feature distributions overlap). Realizable within groups but non-realizable overall.
%\end{enumerate}

In general, $\Phi$ may have multiple fixed points that  demonstrate different characteristics.  We therefore compare the fixed points of $\Phi$ on several metrics of societal importance.

\begin{enumerate}%[leftmargin=*]
	\item {\em Stability:} We say that an equilibrium $\pi^*$ is stable if there is a non-zero measure set of initial states $\pi(0)\in [0,1]$ for which $\pi^* = \lim_{n \rightarrow \infty} \Phi^n(\pi(0))$. In particular, if there exists a neighborhood around $\pi^*$ such that all points converge to $\pi^*$ under the dynamics, we say that $\pi^*$ is \emph{locally stable}. \lledit{As such, stable fixed points are robust to small perturbations in the qualification rate, which can occur due to random measurement errors.}%:= \underbrace{\Phi \circ \Phi \circ \dots \circ \Phi}_{ n \,\,\text{times}}(\pi)$. %, given these are the equilibria that are feasible to obtain. % unstable equilibria are hard to attain. %The stability of the dynamical system implies that there is a class of models or initial conditions for which the trajectories would be equivalent. 

	\item {\em Qualification Rate of Group $a$:} 
	Recall that $\pi_a$ is the fraction of individuals in group $a$ who have decided to invest in qualifications. Since it is more desirable to have a high \investmentlevel~in each group, we may compare equilibria based on $\pi_a$. \nhedit{We refer to $\pi_a = G(w)$ as the \optinvest in group $a$, which is the maximum achievable qualification rate corresponding to the perfect assessment rule.}\footnote{If group $a$ has a group-specific cost distribution, $G_a$, then we refer to $G_a(w)$ as the \optinvest in group $a$.}
	\item {\em Balance:}
	We may be interested in equilibria where the qualification rate is similar across groups, that is, to prioritize equilibria with smaller $\max_{a_1,a_2\in \calA} | \pi_{a_1}^\ast - \pi_{a_2}^\ast|$. When this quantity is $0$ we say that $\pi^*$ is \emph{fully balanced}.
%	given two groups $a_1$ and $a_2$, we might prioritize more balanced equilibria, where we call an equilibria {\em balanced} if $| \pi_{a_1}^\ast - \pi_{a_2}^\ast|=  0$. %and {\em unbalanced} if $| \frac{\pi_{a_1}^\ast}{\pi^\ast} - \frac{\pi_{a_2}^\ast}{\pi^\ast}| =  1$. 
	\item {\em Institutional utility:} We may compare equilibria based on their corresponding institution utility.% to the institution.
	% to that has the highest institutional utility to other equilibria that may be more d we might be interested  understand if the institutional utility maximizing equilibria is alligned with equilibria that are desirable according to other metrics. 
\end{enumerate}
%The proposed model of feedback dynamics is closely related to \citet{coate93will}, with several key distinctions: 1)~We assume the institution has access to multiple features $X$ when learning about the qualifications of each individuals, whereas~\citet{coate93will} assume $X$ a one-dimensional `noisy signal.' 

\subsection{Examples From the Real World}\label{sec:examples}
Let us instantiate our model in the setting of two important applications from the real world.

%We now list some examples of real-word dynamics that fit nicely into our setting. 

\paragraph{College Admissions} Consider the college admission setting, where $X$ corresponds to the features that the college can observe, e.g., a candidate's  test scores and letters of recommendation.
$Y$ indicates whether the candidate meets the qualifications required to succeed in the program. $C$ is the cost of investing in the qualifications,  e.g., the money and opportunity cost of studying or taking additional courses to obtain the required qualifications.
A candidate from group $a$ will develop features from distribution $\Pr[X = x| Y=y, A = a]$, where $y=1$ indicates a qualified candidate. The differences in the feature distribution between groups can be attributed to several factors such as resources that are available to different groups, e.g., letters of recommendations for qualified female and male candidates often emphasize different traits. $\theta$ is the decision parameter used by the college, e.g., $\hat Y_\theta =1$ when the candidate has SAT score of $> 1400$ \lledit{and} an excellent recommendation letter.
%and $Y_\theta$ is the decision rule of the college. 
The college accepts applicants by trading off between the utility gain, $p_{\text{TP}}$, of admitting qualified candidates and utility cost, $c_{\text{FP}}$, of admitting an unqualified candidates. The candidate is incentivized to acquire the qualifications for the college based on the long term
benefit (described in Equation~\eqref{eq:individ_dec}) that depends on their expected gain $w$ from completing a college degree and how likely it is to be admitted to college for a \lledit{qualified or unqualified} member of the group the candidate belongs to.

\paragraph{Hiring} Consider the hiring setting, where $X$ corresponds to the features that the firm can observe, e.g., a candidate's CV.
$Y$ indicates whether the applicant meets the qualifications required by the firm, e.g., having the required knowledge and the ability to work in a team.
$C$ is the cost of acquiring the qualifications required by the firm, e.g., the (monetary and opportunity) cost of acquiring a college degree or working on a team project. 
Parameter $\theta$ is the hiring parameter used by the firm, e.g., $\hat Y_\theta =1$ when the applicant has a software engineering degree and two years of experience. 
The firm accepts candidates according to utility assessment involving $p_{\text{TP}}$, the profit from hiring a qualified candidate, and $c_{\text{FP}}$, the cost of hiring an unqualified candidate e.g., the loss in productivity or the the cost to replace the employee. The candidate, on the other hand, is incentivized to acquire the qualifications for the job based on factors including their expected salary $w$ and how likely it is to be hired by the firm given how the firm has hired qualified or unqualified candidates from the group the candidate belongs to (Equation~\eqref{eq:individ_dec}). 

We also consider a stylized example of lending in Section~\ref{sec:experiments}.
\section{Importance of (Near) Realizability}\label{sec:realizable}
%\nhcomment{Switched $\theta^*$ to $\otheta$ to avoid confusion between optimality and equilibria. It's a macro that can be changed if necessary.}

We start our theoretical investigation of dynamic algorithmic decision making with the classical model of realizability.
In the theory of machine learning, a distribution is called realizable  if there is a decision rule in the set $\Theta$ whose error on the distribution is $0$.
Analogously, we call a setting \emph{realizable}  when there is a decision rule $\otheta \in \Theta$ that perfectly classifies every individual from every group, that is $\tpr_a(\otheta) =1$ and $\fpr_a(\otheta) = 0$ for all $a \in  \calA$.
Realizability is a widely used assumption and is the basis of seminal works such as Boosting~\citep{FS97}.
At a high level, realizability corresponds to the assumption that there is an unknown ground truth \lledit{assessment rule}, for example, in a hypothetical setting where $x$ includes all the information that is sufficient for assessing one's qualification, \lledit{and the chosen set of decision rules is rich enough to contain it}.

%
%hiring a \emph{print model} would be realizable in a hypothetical setting where if the qualification of a candidate with portfolio $x$ --- including e.g., previous photos, height, size, gender, etc. --- can be entirely assessed by examining the portfolio $x$.
%}

In static realizable applications of machine learning, the goal is to (approximately) recover $\otheta$ from data.  We show that in the our dynamic setting, under realizability, the 
unique non-zero equilibrium of $\Phi$ is where individuals respond to $\otheta$. Furthermore, each group attains their optimal qualification rate at this equilibrium.

%\nhcomment{Need citation? or our theorem}
\begin{proposition}[Perfect classfication]
\label{prop:realize}
If there exists $\theta \in \Theta$ such that $\tpr_a(\theta)=1$ and $\fpr_a(\theta) = 0$ for all $a\in \calA$,
then there is a unique non-zero equilibrium with $\pi^*_a = G(w)$ for all $a\in \calA$.
%\footnote{Note that $\pi^*_a = 0$ for all $a\in\calA$ is a zero equilibrium that corresponds to decision rule that accepts no one.}
\end{proposition}

While realizability is a common assumption in the theory of machine learning, it rarely captures the subtleties that exist in automated decision making in practice. Next, we consider a mild relaxation of realizability and consider a setting where a near-perfect  decision rule $\theta\in \Theta$ exists such that $\tpr_a(\theta) \geq 1-\epsilon$ and $\fpr_a(\theta)\leq \epsilon$. As we show (and prove in Appendix~\ref{app:theoremproof}), when there is a single near-realizable group the main message of Proposition~\ref{prop:realize} remains effectively the same. That is, all equilibria that are reachable from initial points that are not too extreme approximately maximize \lledit{the group's} \investmentlevel.

%While the above result discusses the equilibrium of the dynamical system on a given set of individuals, it is not hard to see that it can generalize, in the Probably Approximate Correct (PAC) sense, to the underlying distribution that those individuals belong to. This is due to the fact that in the realizable setting, any $\theta$ with $\tpr(\theta)=1$ and $\fpr(\theta)=0$ on an i.i.d. set of $\Omega\left(\mathrm{VCDim}(\Theta)/\epsilon\right)$ has $\tpr(\theta)> 1-\epsilon$ and $\fpr(\theta)< \epsilon$ on the underlying distribution. Therefore    
%\todo{change to proposition.}
%
%\subsection{Equilibria under approximate realizability}

%\nhdelete{The following theorem shows that if there exists a classifier in the model class that has low false positive and false negative error, and the starting level of investment in the population is not too extreme, then any equilibrium level of investment must be close to the optimal level, $G(w)$. The former assumption is a relaxation of realizability by the model class, and is possible to check. The theorem also requires mild regularity conditions on $G$, the CDF of investment costs.}

\begin{theorem}[Equilibria under near-realizability]\label{thm:eqi_approx_real}
Let $|\calA | = 1$ and assume that $\ptp = \cfp = 1$. 
Assume that for fixed $\epsilon \in (0,1)$ , $s \in (0,1/2)$, $G$ is $L_G$-Lipschitz with property that $1-s \ge G(w) \ge s+ \frac{L_G w \epsilon}{s}$, and  there is $\theta\in \Theta$ such that 
\[ \tpr(\theta) \geq 1-\epsilon \text{ and } \fpr(\theta) \leq \epsilon.\] 
%	Consider a single group. Assume $\ptp = \cfp = 1$. Fix $\epsilon \in (0,1)$ and $s \in (0,1/2)$. Suppose there exists $\theta_0 \in \Theta$ satisfying \[ \tpr(\theta_0) \geq 1-\epsilon \text{ and } \fpr(\theta_0) \leq \epsilon, \] and $G$ is $L_G$-Lipschitz such that $1-s \ge G(w) \ge s+ \frac{L_G w \epsilon}{s}$.
Then for any initial investment $\pi(0) \in [s,1-s]$,  $\pi^* = \lim_{n\rightarrow\infty} \Phi^{n}(\pi(0))$ is such that \[\pi^* \geq G(w(1-\epsilon/s)).\]
\end{theorem}
%\nhcomment{The theorem said we have ``at least $\pi^* = ...$'', I replaced it with $\pi^* \geq ...$. Is that what you meant?}

A nice aspect of the above results is that the assumption of realizability or near-realizability can be validated from the data. That is, the decision maker can compute whether there is $\theta\in \Theta$ such that $\tpr(\theta)\geq 1-\epsilon$ and $\fpr(\theta)\leq \epsilon$. If so, then the decision maker can rest assured that the dynamical system is on the path towards achieving near optimal investment by the individuals. Another nice aspect of these results is the characterization of the equilibria in terms of the CDF of the cost distribution. This allows us to use this framework for studying interventions that change the cost function directly. \lledit{One such intervention is subsidizing the cost for individuals so that the cumulative distribution function of the cost, given by $G(x)$, is increased by a sufficient amount at every cost level $x$. The following corollary, proved in Appendix \ref{app:corproof}, shows that under this kind of subsidy, the equilibria reached by the dynamics will have higher \investmentlevel~than any fixed point of the dynamics before subsidy, as long as the initial points are not too extreme. As we are considering different cost distribution functions in the following corollary, we denote the dynamics corresponding to cost distribution function $G$ as $\Phi_G$.}

%\todo{move the proof}
\begin{corollary}[Subsidizing the cost of investment] \label{cor:subsidy}%Under the conditions established in Theorem~\ref{thm:eqi_approx_real}, if costs of investment are subsidized resulting in $G'$ stochastically dominates the orignal CDF of the cost of investment $G$, then the set of new equilibria $\pi({G'})$ will stochastically dominate the set of old equlibria $\pi({G})$.
Let $|\calA | = 1$ and assume that $\ptp = \cfp = 1$. 
Assume that for fixed $\epsilon \in (0,1)$ , $s \in (0,1/2)$, $G$ is $L_G$-Lipschitz with property that $1-s \ge G(w) \ge s+ \frac{L_G w \epsilon}{s}$, and there is $\theta\in \Theta$ such that 
\[ \tpr(\theta) \geq 1-\epsilon \text{ and } \fpr(\theta) \leq \epsilon.\] 
Let $\pi^\ast >0$ be a fixed point of the dynamics $\Phi_G$. Suppose $\bar G$ is a strictly increasing, $L_{\bar G}$-Lipschitz CDF such that $1 - s \geq \bar G(x) \ge s+ \frac{L_{\bar G} w \epsilon}{s}$ and  $\bar G(x(1- \epsilon/s)) \geq G(x)$ for all $x$ in the domain of $G$. %and $L_{\bar G}\leq L_{G}$.  
Then for any initial investment $\pi(0) \in [s,1-s]$, there exists a $\bar \pi \geq \pi^\ast$, such that $\bar \pi = \lim_{n\rightarrow\infty} \Phi_{\bar G}^{n}(\pi(0))$.

\end{corollary}

%\nhcomment{Can you write what exactly the requirements are. Because Theorem 3.2 assumption uses $L_G$ so I'm not sure if we need to make these assumptions on $L_{G'}$ or if we should change $s$, or something else. So it's better to be careful here and spell out exactly what we get in this corollary}

% !TeX root = main.tex 

\section{%Relaxing Realizability with Heterogeneous Groups
Group Realizability}\label{sec:groups}

In this section, we investigate how the nature of equilibria evolves as the assumption of realizability is relaxed to allow for heterogeneity across groups. Specifically, we consider the case where there exists a perfect \drule~for each group, but not when the groups are combined. We call this ``group-realizability". \lledit{On a high level, our results illustrate that without realizability or near-realizability, the utility-maximizing \drule~can be very sensitive to the relative qualification rates in different groups, leading to the existence of \emph{multiple} equilibria --- at some of which groups experience  disparate outcomes.}

\nhedit{In sections~\ref{sec:uniform} and \ref{sec:multivariate-gaussian}, we study group-realizability under} two different and complementary settings.
%\nhedit{--- including the model setup for  distributions $\Pr[X=x | Y=y, A=a]$ for groups $a\in \A$ and qualification $y\in \calY$.}
The first setting considers features that are drawn from a multivariate Gaussian distribution and assumes that in each group the qualified individuals are perfectly separated from unqualified ones by a group-specific hyperplane. This is a benign setting where no group is inherently disadvantaged --- group features and performance of \drules~are symmetric up to a reparameterization of the space. The second setting considers features that are uniformly distributed scalar scores and assumes that qualified and unqualified individuals in a group are separated by a group-specific threshold, where one is higher than the other.
\nhedit{
	%The first setting considers group features ($\Pr[X=x| A = a]$) that form a multivariate Gaussian distribution and assumes that in each group the qualified individuals ($\Pr[X=x| Y = 1, A = a]$) are perfectly separated from unqualified ones  ($\Pr[X=x| Y = 0, A = a]$) by a group-specific hyperplane. 
  -%This model captures a benign setting where there are no inherent differences between two groups --- group features and performance of \drules are symmetric up to a reparameterization of the space.
%The second setting considers group features to be uniformly distributed scalar scores and assumes that qualified and unqualified individuals in a group are separated by a group-specific threshold, where one is naturally higher than others.
}
This model captures the natural setting where the feature (score) and \drules~inherently favor one group, e.g., SAT scores are known to be skewed by race \citep{card07racial}. % \nhedit{feature (score) and assessment rules inherently favor one group, e.g., SAT scores are known to be higher on average for more affluent applicants.}
\nhedit{We use the aforementioned} stylized \lledit{settings} to demonstrate the salient characteristics of equilibria that one might anticipate \lledit{under group-realizability}%in the group-realizable setting
. We find that stable equilibria tend to favor one group or the other. \nhedit{This is especially surprising in the multivariate Gaussian case where the two groups are  identical up to a change in the representation of the space.}
\nhedit{We also study the existence of balanced equilibria, where both groups acquire qualification at the same rate. We find that when balanced equilibria exist they tend} to be unstable, that is, \nhedit{no initial qualification rate (except for the balanced equilibrium itself) will converge to the balanced equilibrium under the dynamics.}

\nhedit{We consider two natural interventions in overcoming the challenges of group-realizability as outlined above. 
	As group-realizability poses even greater challenges when the costs of investment are unequally distributed \emph{between} groups, in Section~\ref{sec:subsidies} we consider the impact of subsidizing the cost of acquiring qualification for one group.
In Section~\ref{sec:decoupling}, we consider the impact of decoupling, that is, we allow the institution to use different assessment rules} for different groups assuming the group attributes are available. This is in contrast to the typical setting where institutions are constrained to using the same assessment rule across all groups, which may be the case when
data on the protected attribute is not available or when the use of protected attributes for assessment is regulated.
%\nhedit{This is in contrast to the typical setting where institutions are constrained to using the same assessment rule across all groups, which may be due to the fact that individual's protected attribute may not be available or the use and elicitation of protected attributes for assessment may be legally forbidden.}

%In practice this could be due to legislation that forbids the elicitation and use of protected attributes or simply the lack of such information. 
% !TeX root = main.tex 

\subsection{Uniformly Distributed Scalar $X$}\label{sec:uniform}
\nhedit{We consider $\calX = [0,1]$, the class of assessment paramters $\Theta = [0,1]$, and assessment decision $\hat{Y}_h = \Ind{X> h}$ for all $h\in \Theta$ that represent all threshold decision policies. 
}
Consider two groups $a_1, a_2$. \nhedit{We consider $X$ to be a score that is uniformly distributed over $[0,1]$ where in group $a_i$ those with score more than $h_i$ are qualified and those with score at most $h_i$ are unqualified. \lledit{This is depicted in Figure~\ref{fig:models} (right).} Formally, we have
\begin{align*}
\Pr(X = x \mid Y = y, A = a_i) = 
\begin{cases}
\Ind{x > h_i} / (1- h_i) & \text{for } y=1 \text{ and}\\
\Ind{x \leq h_i} / h_i & \text{for } y=0\\
\end{cases}
\end{align*}
}

\begin{figure}[t]
	\setlength{\belowcaptionskip}{-6pt}
	\vspace{-1em}
	\centering
	\includegraphics[width=0.45\columnwidth, trim = 0 -30 0 0, clip]{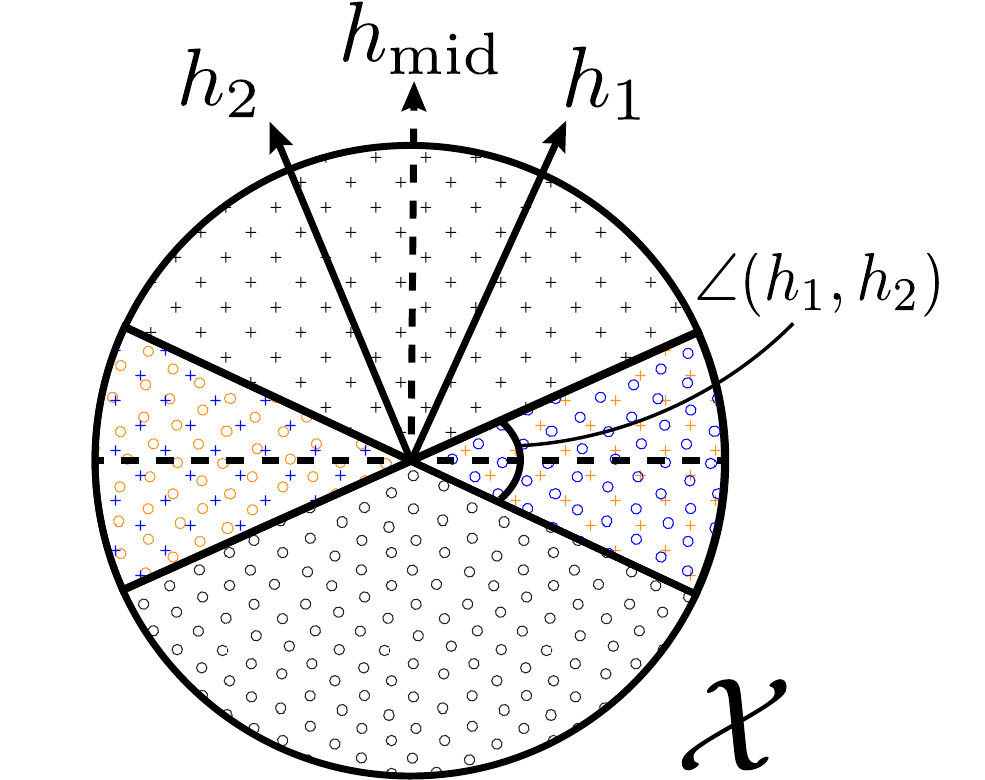}%
	\includegraphics[width=0.35\columnwidth, trim = -10 10 20 0, clip]{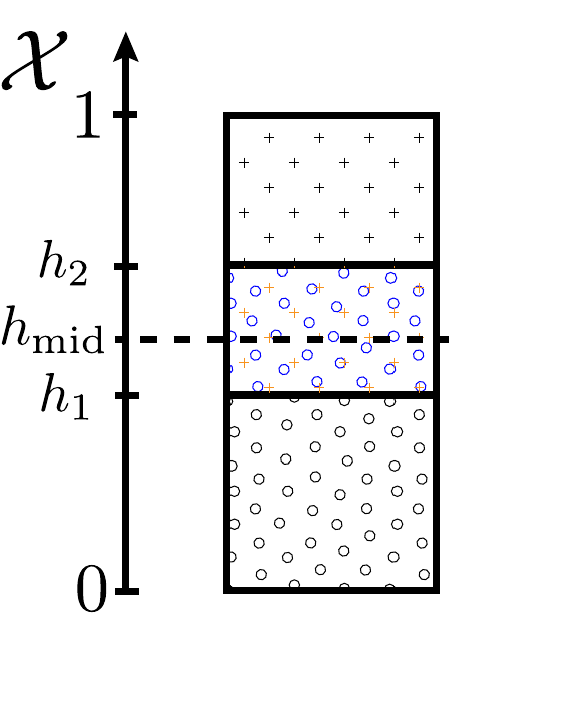}
	\includegraphics[width=0.5\columnwidth, trim = 0 0 0 0, clip]{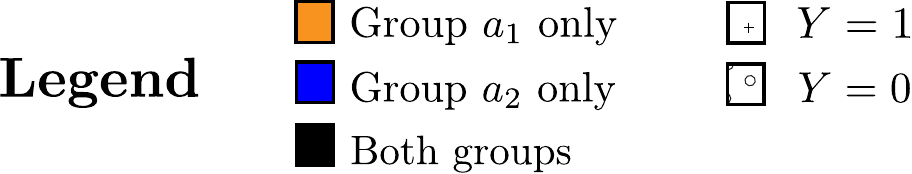}
	\caption{Equilibria in Multivariate Gaussian case (left) Uniform case (right)}\label{fig:models}
\end{figure}

%\nhcomment{I replaced the equation to use the format of $\Pr[X=x]$ that we have been using so far.} 
%\todo{I change the equation above and since $(\cdot)_+$ wasn't needed as part of it I removed it. If $(\cdot)_+$ is used elsewhere please make sure it's introduced there.}
%	where $(\cdot)_+$ denotes the positive part and $0 < \max( h_1, 1-h_1) < h_2 < 1$. We consider threshold decision policies $\hat{Y}_h = \Ind{X > h}$ and $h\in [0,1]$ is the class of policy parameters.
	\begin{assumption}\label{ass:unif}We assume $n_{a_1} \cdot \ptp = n_{a_2}  \cdot \cfp$. \lledit{We also assume, for simplicity that the cost for acquiring qualifications is uniformly distributed on $[0,1]$ (i.e. $G(c) = c$) in both groups}; our results, however, generalize to the setting where the CDF for the cost $G:[0,1]\to [0,1]$ is an arbitrary strictly increasing function. 
	\end{assumption}
\nhedit{We show that} \awedit{when $w$ is in a certain range}, %similarly as in the Gaussian setting,  
there are two \nhedit{\emph{unbalanced stable}} equilibria corresponding to assessment parameters $h_1$ or $h_2$, which respectively lead to the \optinvest for groups $a_1$ or $a_2$ but low qualification rate for the other group.
%When wage $w$ is in a specific range,} 
\awedit{There is also a more} \nhedit{\emph{balanced} but \emph{unstable} equilibrium at some threshold $\hm$ between $h_1$ and $h_2$. 
Outside of this range of $w$, there is only one equilibrium in which one of the groups achieves its \optinvest. These findings are summarized in} the following two propositions.

%	\todo{Let's discuss this in person ore to clarify where iff part is affecting.}
	\begin{proposition}\label{prop:uniform-eq} Define $g \coloneq
	\frac{(1-h_1)(-wh_2^2+h_2(1-h_1)-wh_1(1-h_1))}{w((1-h_1)^2-h_2^2))}.$ Note that $g\in (0, h_2-h_1)$ for any $w$. Let $w \in \left(w_l, w_u \right)$ where 
	\begin{align}\label{eq:w-bounds}
	%w \in \left(w_l, w_u \right), \text{ where } 
	w_{l}:=\frac{(1-h_1)^2}{(1-h_2)h_2+(1-h_1)^2},~ w_u := \frac{h_2(1-h_1)}{h_2^2+h_1(1-h_1)}.
	\end{align}
	Given Assumption~\ref{ass:unif}, there exists \nhedit{two} stable equilibria at
\nhedit{
	\begin{align}
	h&=h_1, & \pi_{a_1} &= w, &   \pi_{a_2} &= w \cdot \frac{h_1}{h_2} , \text{ and} \label{eq:h1-uniform} \\
	h &= h_2, & \pi_{a_1} &= w\cdot \frac{1-h_2}{1-h_1}, & \pi_{a_2} &= w, \label{eq:h2-uniform}.
	\end{align}
	}
	and \nhedit{a unique non-zero} unstable equilibrium at
	\begin{equation*}
	h=\hm :=h_1 + g,\quad	\pi_{a_1} = w\cdot \frac{1-h_1-g}{1-h_1},\quad \pi_{a_2} =w\cdot\frac{h_1+g}{h_2}. 
	\end{equation*}
	In addition, when $w = 1 - h_1$ the unstable equilibrium is fully balanced. 
	\end{proposition}

\begin{proposition}
Given Assumption~\ref{ass:unif} when $w <w_l$ there exists one stable equilibrium \nhedit{defined by Equation~\ref{eq:h2-uniform}}, and when $w >w_u$ there exists one stable equilibrium \nhedit{defined by Equation~\ref{eq:h1-uniform}.}
\end{proposition}
The details of the proofs are \nhedit{presented} in Appendix \ref{app:proof-uniform}.
\nhedit{At a high level,}  %if the wage is too low, the incentive to invest is low and the utility-maximizing $h$ is determined by the group $a_2$ with a lower number of positives $h_2$.\todo{reword ``lower number of positives $h_2$.''} If the wage is too high, the incentive to invest is high and the utility-maximizing $h$ is determined by the group $a_1$ with a higher fraction of positives $h_1$\todo{reword higher fraction of positives $h_1$}. I
if the wage is not too low or too high, \lledit{both thresholds $h_1$ and $h_2$ correspond to stable equilibria, at which either group $a_1$ or $a_2$ is perfectly classified.} The equilibrium corresponding to $\hm$, where the classifier has the same true positive and false positive rates in both groups, is unstable and subsequently harder to achieve. 

\lledit{In Table~\ref{tab:uniform-short}, we compare these equilibria in terms of metrics introduced in Section~\ref{sec:model}, under the assumptions of Proposition~\ref{prop:uniform-eq}. We use standard notation $\succ$ and $\sim$ to denote preference and indifference respectively. For example, %in row 4 of Table~\ref{tab:gaussian-short}, 
we find that in terms of balance in qualification rates, the stable equilibrium associated with $h_1$ is more balanced that the stable equilibrium associated with $h_2$, but both are always less balanced unstable equilibrium associated with $\hm$. }Details of the computation are deferred to Table \ref{tab:uniform} in Appendix \ref{app:proof-uniform}.
\begin{table}[htbp!]
	\begin{center}
		\scalebox{0.8}{
		\begin{tabular}{| p{5cm} |  p{4.2cm}|} 
			\hline
			&  \thead{Ranking of Equilibria}  \\ [0.5ex] 
			\hline
			\makecell{Stability\\\quad}  & \makecell{$h_1, h_2$ are stable. \\$\hm$ is unstable} \\ 
			\hline
			\makecell{Qualification rate of group $a_1$}   & $\makecell{h_1 \succ \hm \succ h_2}$  \\ 
			\hline
			\makecell{Qualification rate of group $a_2$}   & $\makecell{h_2 \succ \hm\succ h_1}$   \\ 
			\hline
			\makecell{Balance of qualification rates } & $\makecell{\hm\succ h_1 \succ h_2}$  \\ 
		
			\hline
			\makecell{Institution's Utility} &  \makecell{no ranking} \\ [.5ex] 
			\hline
		\end{tabular}}
	\end{center}
	\caption{Comparison of equilibria for uniform features. In this table we refer to each equlibria using the associated threshold decision policy.}\label{tab:uniform-short}
\end{table}
% !TeX root = main.tex 
\subsection{Multivariate Gaussian $X$} \label{sec:multivariate-gaussian}

\awedit{We consider $\calX = \R^d$ and
$\Theta = S_{d-1}$, where $S_{d-1}$ is the set of $d$-dimensional unit vectors. Let $\hat{Y}_h = \Ind{X\trans h \geq 0}$ for all $h\in \Theta$ denote separating hyperplane policies  and} \nhedit{ $\angle_{h,h'} \coloneqq \frac{1}{\pi}\arccos(\frac{h\trans h'}{\norm{h}\norm{h'}})$} denote the angle between two vectors, normalized by the constant~$\pi$.
\nhedit{We consider two groups $a_1$ and $a_2$ associated respectively with vectors $h_1$ and $h_2$, such that $\angle_{h_1, h_2} \neq 0$. We assume the groups have equal size, i.e., $n_{a_1} = n_{a_2}$.
For each group the feature distribution forms a $d$-dimensional spherical  Gaussian centered at the origin such that the qualified individuals are in halfspace $\Ind{X\trans h_i \geq 0}$ and the unqualified individuals in halfspace $\Ind{X\trans h_i < 0}$. Formally,} for $x \in \R^d$ and  $i \in \{1,2\}$,
\nhedit{
\begin{align*}
	\Pr(X = x \mid Y=y, A=a_i) = \begin{cases}
	2\phi(x)  \Ind{x\trans h_i \ge 0} & \text{for } y=1 \text{ and}\\
	2\phi(x)  \Ind{x\trans h_i < 0} & \text{for } y=0,
	\end{cases}
\end{align*} 
where $\phi(x)$ is the density of the spherical $d$-dimensional Gaussian. \lledit{This is depicted in Figure~\ref{fig:models}~(left).}
}
%	We consider the class of separating hyperplane policies $\hat{Y}_h = \Ind{X\trans h \geq 0}$ and $h\in \B^d_0$ is the class of policy parameters.	
	%Assume cost of skill acquisition, $C$, is uniformly distributed on $[0,1]$. 
\begin{assumption}\label{assmp:mvgaussian}
			We assume that the CDF for the cost of \lledit{acquiring qualifications} is a strictly increasing function $G:[0,1]\to [0,1]$ and is  \lledit{the same in both groups}.
%			 $w\in (0,1)$ and $G:[0,1]\to [0,1]$ the strictly increasing CDF for cost of investment, which is the same for both groups.
\end{assumption}

\nhedit{As we will see, the relative gain and loss of the institution for respectively accepting a qualified or unqualified individual, that is $\ptp / \cfp$,  plays a role in the nature of the equilibria.
The following proposition characterizes the equilibria when this value is strictly positive, that is,} when the benefit of true positives outweighs the cost of false positives. \lledit{Surprisingly similar to the previous setting of uniform scores, the current setting also has two stable equilibria that each favor one group at the expense of the other, as well as a balanced equilibrium that is unstable.}

\begin{proposition}\label{prop:gaussian-eq}
	Given Assumption \ref{assmp:mvgaussian} and
$\ptp > \cfp$, there exists two stable equilibria, at
\begin{align*}
	h &=h_1,  &     \pi_{a_1} &= G(w)   & \pi_{a_2} &= G\left( w \cdot (1-2 \angle_{h_1,h_2} )\right) , \\
	h &= h_2, &     \pi_{a_1} &= G\left(w (1-2 \angle_{h_1,h_2} )\right)   & \pi_{a_2} &= G(w).
\end{align*}
There is a unique non-zero unstable equilibrium at 
\begin{align*}
	h &= \hm , &  \pi_{a_1} &=  G\left(w(1- \angle_{h_{1},h_2} )\right) & \pi_{a_2} &= G\left(w(1-\angle_{h_{1},h_2})\right),
\end{align*}
where $\hm := \frac{h_1+h_2}{\norm{h_1+h_2}}$.
\end{proposition}

\nhedit{Let us briefly comment on the high level proof idea} and defer the full argument to Appendix \ref{app:proof-mvgaussians}. 
Since $\ptp > \cfp$, the institution cares more about accepting true positives than avoiding false positives.
Therefore, the utility-maximizing $h$ is determined by the group that has \nhedit{a higher qualification rate} and thus has a higher fraction of positives --- this is $h_1$ (resp. $h_2$) whenever $\pi_{a_1} > \pi_{a_2}$ (resp. $\pi_{a_1} < \pi_{a_2}$). When qualification rates  are equal between the two groups, the institution maximizes its utility at any $h$ that is a convex combination of $h_1$ and $h_2$, but the unique $h$ that would induce equal \nhedit{qualification rate} is $h = \hm$, where the classifier has the same true positive and false positive rates in both groups.

An unfortunate implication of this result is that the dynamics will always converge to an unbalanced \nhedit{qualification rate}, except when the initial levels of investment are exactly the same. Even though a fully balanced equilibrium exists, it is unstable and therefore not robust to small \awedit{perturbations in either the \nhedit{qualification rates} or the classifier}, which in practice is unavoidable given sampling noise.

In Table~\ref{tab:gaussian-short}, we compare these equilibria in terms of metrics introduced in Section~\ref{sec:model}. %We use standard notation $\succ$ and $\sim$ to denote preference and indifference respectively. 
For example, %in row 4 of Table~\ref{tab:gaussian-short}, 
we find that in terms of institutional utility, the stable equilibria associated with $h_1$ and $h_2$ are equally preferred, and are both strictly preferred to the unstable equilibrium associated with $\hm$. This implies that the institution has no incentive at all to keep the dynamics at the unstable equilibrium, even though it induces balanced investment. Exact values are deferred to Table~\ref{app:tab:gaussian} in Appendix~\ref{app:proof-mvgaussians}.

\begin{table}[htbp]
	\begin{center}
	\scalebox{0.8}{
		\begin{tabular}{| p{5cm} |  p{4.2cm}|} 
			\hline
			&  \thead{Ranking of Equilibria}  \\ [0.5ex] 
			\hline
			\makecell{Stability\\\quad}  & \makecell{$h_1, h_2$ are stable. \\$\hm$ is unstable} \\ 
			\hline
			\makecell{Qualification rate of group $a_1$}   & $\makecell{h_1 \succ \hm \succ h_2}$  \\ 
			\hline
			\makecell{Qualification rate of group $a_2$}   & $\makecell{h_2 \succ \hm\succ h_1}$   \\ 
			\hline
			\makecell{Balance of qualification rate} & $\makecell{\hm\succ h_1 \sim h_2}$  \\ 
		
			\hline
			\makecell{Institution's Utility} &  $\makecell{h_1\sim h_2 \succ \hm}$ \\ [.5ex] 
			\hline
		\end{tabular}}
	\end{center}
	\caption{Comparison of equilibria for Multivariate Gaussian features. In this table we refer to each equlibria using the associated hyperplane.}\label{tab:gaussian-short}
\end{table}
%\vspace{-2em}

\awedit{Interestingly, when $\ptp<\cfp$, there are no stable equilibria; instead there exists a stable limit cycle between $h_1$ and $h_2$. This is stated informally in the following proposition.
\begin{proposition}\label{prop:gaussian-limit}
	Given Assumption \ref{assmp:mvgaussian} and
	\nhedit{$\ptp < \cfp$,} there exists no stable equilibria. Instead there exists a limit cycle and one non-trivial unstable equilibrium. 
\end{proposition}
Intuitively, the cycle is caused by misaligned incentives between the institution and the individuals. As the institution finds false positives more costly than false negatives, it prefers the hyperplane that classifies more false positives correctly. At each time step, it will choose the hyperplane associated with the group that has a lower \nhedit{qualification rate}, prompting that group to invest more in the next time step.} \lledit{It is striking that even in an simple model involved multivariate Gaussian distributions, a range of limiting behavior is possible for the dynamics in the group-realizable setting. In Section~\ref{sec:experiments}, we also observe the existence of limit cycles in simulations with real data distributions.}

% !TeX root = main.tex 

\subsection{Different Costs of Investment by Group}\label{sec:subsidies}

Thus far we have assumed that all groups have the same distribution of the cost of investment, $G$. In reality, the cost of investment may be distributed differently in each group; a disadvantaged group might on average experience higher \nhedit{(monetary or opportunity)} costs. 
%For example, art students from households with low \awedit{amounts of} disposable income might find painting classes \nhedit{to be} more expensive than \awedit{would high-income households} as the tuition is a larger proportion of their income and the opportunity cost is prohibitive. 
\nhedit{For example, 
low income families who may have to take out loans to pay for college tuition incur high interest rates.
}
This is a compelling setting that reflects deep-seated disparities in access to opportunity between demographic groups in the real world; an analogous setting has been considered by works on strategic classification, where the costs \nhedit{for} manipulating features is posited to differ across groups \citep{Hu2019disparate,Milli2019social}.

In this section, we consider the ramifications of \nhedit{differences in investment cost} across groups, focusing on the setting of Section~\ref{sec:multivariate-gaussian}. We show that the disadvantage from having higher costs is amplified under group-realizability. \awedit{Specifically, suppose} that group $a_1$ (resp. $a_2$) has costs distributed according to cumulative distribution function $G_1$ (resp. $G_2$), and that group $a_1$ is disadvantaged in terms of costs. The following result observes that if $G_1$ sufficiently dominates $G_2$, then there exists no stable equilibrium that encourages optimal investment from group $a_1$ and no equilibrium that is balanced for both groups, in sharp contrast to the characterization in Proposition~\ref{prop:gaussian-eq}. The proof is deferred to Appendix~\ref{app:proof-subsidies}.

\begin{proposition}\label{prop:unequal_costs_eq}
\nhedit{Consider the multi-variate Gaussian setting of Section~\ref{sec:multivariate-gaussian}.}
Suppose $G_1$ and $G_2$ are such that $G_1(w) < G_2(w(1-2\angle_{h_{1},h_2}))$, then there exists a single non-trivial equilibrium at $h_2$, which is also stable. The level of investment by group $a_1$ (resp. $a_2$) is $G_1(w(1-2\angle_{h_{1},h_2})$ (resp.~$G_2(w)$).
\end{proposition}

\paragraph{Effect of subsidies} In this situation, an intervention that would effectively raise the equilibrium level of investment by the disadvantaged group is to subsidize the cost of investment. In particular, as long as we replace $G_1$ with a stochastically dominated distribution $\bar{G}_1$ such that $\bar{G}_1 > G_2(w(1-2\angle_{h_{1},h_2}))$, under the new dynamics $\Phi\sub$, $h_1$ will again be a stable equilibrium, and there will also exist a more balanced, unstable equilibrium at \nhedit{$h=\bar{h}_\mathrm{mid}$,} which is some convex combination of $h_1$ and $h_2$. At all equilibria of $\Phi\sub$, group $a_1$ will have higher levels of investment than $G_1(w(1-2\angle_{h_{1},h_2}))$.

However, this improvement may come at a cost to the advantaged group, since $\Phi\sub$ has multiple equilibria and some of them have group $a_2$ investing less than $G_2(w)$. Still one might argue that the equilibria of $\Phi\sub$ are more equitable, since the dynamics without subsidies always result in optimal investment by group $a_2$ and low investment by group $a_1$.

% !TeX root = main.tex 
\subsection{Decoupling the Assessment Rule by Group}\label{sec:decoupling}

The models we studied in Sections~\ref{sec:multivariate-gaussian} and~\ref{sec:uniform} suggest that applying the same, or ``joint", \nhedit{assessment  rule} to heterogeneous groups results in undesirable trade-offs---between balance, stability, and other metrics---at all equilibria, even though \nhedit{there is a perfect assessment rule in each group that leads to the \optinvest in that group.}
%classification problem is realizable for each group separately.

Decoupling the classifier by group is a natural intervention in this setting. Namely, the institution may choose a group-specific $\theta_a \in \Theta$ to assess individuals from group $a\in \calA$, assuming that the group attribute information is available. This corresponds to choosing $\theta_a$ that maximizes the utility that the institution derives from each group separately.
Thus we now consider the \emph{decoupled dynamics} $\Phi\dec$ where the institution uses group-specific assessment rules, i.e.,  for all $a\in \calA$
\begin{equation}\label{eq:instit_decoupled}
\brtheta_a(\pi_a) \coloneqq \argmax_{\theta_a\in \Theta} ~ \ptp \tpr_a(\theta_a) \pi_a - \cfp \fpr_a(\theta_a) (1-\pi_a).\footnote{As when we defined the joint dynamics (Section~\ref{sec:model}), when the $\argmax$ is not unique, we assume ties are broken according to a fixed and well-defined order.}
\end{equation}

As in the standard joint setting individuals still acquire qualification according to their group utility as follows
\begin{equation*}
\brpi_a(\theta_a) \coloneqq G(w (\text{TPR}_a(\theta_a) - \text{FPR}_a(\theta_a) )).
\end{equation*}
We denote by $\pi\dec \in [0,1]^{|\calA|}$ the equilibria of the decoupled dynamics, $\Phi\dec = \left(\brpi_a \circ \brtheta_a\right)_{a\in \A}$.
It is not hard to see that decoupling is helpful in a group-realizable setting.
That is, the qualification rates of the decoupled equilibrium $\pi\dec$ \emph{Pareto-dominates} the qualification rates of all equilibria $\pi$ under a joint \nhedit{assessment rule,} whenever group-realizability holds.

\begin{proposition}[Decoupling]
Consider a group-realizable setting, that is, for every $a\in \calA$, there exists a perfect assessment rule $\otheta_a \in \Theta$ such that $\tpr_a(\otheta_a) - \fpr_a(\otheta_a) = 1$.
Then $\Phi\dec$ has a unique stable equilibrium $\pi\dec$, where $\pi\dec_a = G(w)$.
Moreover, for any equilibrium $\pi$ of the joint dynamics $\Phi$, \nhedit{$\pi\dec_a \geq  \pi_a$ for all $a\in \calA.$}
\nhedit{Furthermore, if there is no perfect assessment rule, i.e., 
\[\max_{\theta \in \Theta} \sum_{a\in \calA} n_a (\tpr_a(\theta) - \fpr_a(\theta)) < 1,
\]
then for some $a\in \A$, $\pi\dec_a > \pi_a$.
}
\end{proposition}
The proof of this proposition directly follows from Proposition~\ref{prop:realize}.

Indeed, decoupling always helps in the group-realizable setting---not only does it not decrease any group's \lledit{equilibrium} qualification rate, it also increases the \lledit{equilibrium} qualification rate of at least one group when realizability across all groups does not hold.
In Sections~\ref{sec:non-realizable} and \ref{sec:experiments} we take a closer look at decoupling in the absence of group-realizability and see that those cases are not as clear-cut.
When group-realizability does not hold, we see that in some cases decoupling is still helpful while in others it can significantly harm one group.

%We shall see in Section~\ref{sec:experiments} and~\ref{sec:non-realizable} that this might be true if realizability does not hold within group---some groups could have higher investment levels under the joint dynamics.

% !TeX root = main.tex 

\section{Beyond group-realizability: Multiple equilibria within group}\label{sec:non-realizable}

Thus far we have considered settings where the learning problem is realizable (or almost realizable) within each group. This is a common assumption in various prior works, such as in \citet{Hu2019disparate}.
\nhedit{As we saw in Section~\ref{sec:groups}, there may be multiple undesirable equilibria when a joint assessment rule is used in a group-realizable setting, but these undesirable equilibria disappear in the decoupled dynamics.}

In many application domains, realizability does not hold even \nhedit{at a group level}.
That is to say, no \nhedit{assessment rule in $\Theta$} can perfectly separate \nhedit{qualified and unqualified individuals even within one group.}
This may be due to \nhedit{the fact that mapping individuals to the visible  feature space $\calX$ involves loss of information or there may be other sources of stochasticity in the domain~\citep{corbett18measure}, making it impossible to provide a high accuracy assessment of individuals' qualifications.}
A key consequence \nhedit{of the lack of realizability }is that even for a single group, the optimal classifier \nhedit{now can} vary greatly with $\pi_a$, the \nhedit{group's qualification rate.
As a result, our guarantees about the near-optimality} of stable equilibria (Theorem \ref{thm:eqi_approx_real}) no longer hold, and there could exist multiple stable equilibria each corresponding to a different qualification rate within a group. In this section, we investigate the existence of bad equilibria for a single group and its implications on decoupling when the learning problem is not group-realizable. For the rest of this section, we consider a single group, i.e.,  $|\calA|=1$ and suppress $a$ in the notation.

In the following proposition, we characterize conditions under which multiple equilibria exists in arbitrary feature spaces and assessment rules. This is a generalization of a classical result from \citet{coate93will} that considers a one-dimensional feature space; \lledit{for completeness, we restate and prove this result as a consequence of Proposition~\ref{prop:multi_eq} in Appendix \ref{app:non-realizable}.}
%\todo{$Y_\theta$ was changed to $\hat Y_\theta$, check appendices for consistency.}
\begin{proposition}[Multiple equilibria in arbitrary feature spaces]\label{prop:multi_eq} Let $\Phi$ be as defined in Section~\ref{sec:model}. 
%Fix the group $a$ and suppress $a$ in the notation.	
	%Let $r:= p_{\text{TP}} / c_{\text{FP}}$. \nhedit{Recall that $\Phi = ....$ defined the dynamic system.}
	\nhedit{
		For any qualification rate $\pi$, let
		\[ \beta(\pi) \coloneqq \tpr(\brtheta(\pi)) - \fpr(\brtheta(\pi)),
		\]
		be the difference between true and false positive rates of the institution's utility maximizing assessment rule with respect to $\pi$.
	}
	%	\begin{align*}
	%		\Phi(\pi) &= G(w\beta(\pi))\\
	%		\beta(\pi) &= \text{TPR}_a(\theta^*(\pi)) - \text{FPR}_a(\theta^*(\pi)) \\
	%		\theta^*(\pi) &= \argmax_{\theta\in \Theta} ~r\cdot \text{TPR}_a(\theta) \cdot \pi -  \text{FPR}_a(\theta) \cdot (1-\pi) 
	%		\end{align*}
	Assume $\beta(\pi)$ is continuous, the CDF of the cost $G$ is continuous and that there exists $\theta \in \Theta$ such that $\Pr(\hat Y_\theta = 1) = 0$ and $\theta' \in \Theta$ such that $\Pr(\hat Y_{\theta'} = 1) =1$, \nhedit{i.e., there is a assessment rule that accepts everyone and an assessment rule that rejects everyone.} Also suppose the likelihood ratio \nhedit{$\phi(x) := \frac{\Pr(X = x\mid Y=0)}{\Pr(X = x\mid Y=1)}$} is strictly positive on $\calX$.
	
	If  $x< G(w\beta(x))$ for some $x \in (0,1)$, then there exists at least two distinct non-zero equilibria where $\pi = \Phi(\pi)$. If in addition $\beta$ is differentiable, an equilibrium at $\pi$ is locally stable whenever $G'(w\beta(\pi)) < |\beta'(\pi)|$, where $G'$ and $\beta'$ denote the derivatives of $G$ and $\beta$ respectively.
	% $|G'(w\beta(\pi))| < |\beta'(\pi)|$. 
	
\end{proposition}
\begin{proof}
	When $\pi = 1$, the institution's best response is to accept everyone regardless of their features, so $\beta(1) =  \tpr(\theta') - \fpr(\theta') = 1- 1= 0$.\footnote{Recall that $\fpr(\theta):= \Pr(\hat{Y}_\theta = 1\mid do(Y=0))$, which is well-defined even when $\Pr(Y=0) = 0$.}
	
Note that $\frac{1-\pi}{\pi} 	\rightarrow \infty$ as $\pi \rightarrow 0$. This, together with the fact that  \lledit{$\phi(x)$ is strictly positive} means that	 there must exist $\bar{\pi} > 0$ such that 
	$\frac{\ptp}{\cfp} < \frac{1-\bar{\pi}}{\bar{\pi}}\phi(x)$ for all $x \in \calX.$ Therefore, for all $\pi \le \bar{\pi}$, the institution's best response is to accept no one regardless of their features, so we have that $\beta(\pi) = \tpr(\theta) - \fpr(\theta) = 0-0 = 0$ for $\pi \leq \bar{\pi} $.	

Since $G(0) = 0$, we have that $\bar{\pi} > G(w\beta(\bar{\pi} )  )= 0$ and $1 > G(w\beta(1 ))  = 0$.
	 By assumption there exists $x < G(w\beta(x))$ for some $x \in (0,1)$, and by the above discussion, we must have $x \in (\bar{\pi},1)$. Hence there must be at least 2 solutions to $\pi = \Phi(\pi)$ in  $(\bar{\pi},1)$ and in particular they are non-zero. The condition for local stability follows directly from chain rule.
\end{proof}

Proposition~\ref{prop:multi_eq} describes conditions under which there exists more than one equilibrium in the dynamics modeled in Section 2. Given a differentiable $\beta(\pi)$, one can always construct a monotonically increasing $G$, such that the dynamics $\Phi$ has any number of locally stable equilibria. The implication of having multiple equilibria is that the dynamics may converge to different equilibrium \nhedit{qualification rates} depending on the initial investment, even for a single group. This makes the setting particularly hard to analyze. %This has 

Nevertheless, the following result shows that even in the non-realizable setting, subsidizing the cost of investment by changing the distribution $G$ to a stochastically dominant distribution $\bar{G}$ will create a new equilibrium that has a higher \investmentlevel. In other words, subsidies in the non-realizable setting also improve the quality of equilibria. However, the new equilibrium is not guaranteed to be locally stable. We see some ramifications of this empirically in the next section.
 
\begin{proposition}[Subsidies without realizability]\label{prop:non-realizable-subsidies}
Suppose $\pi^* > 0$ is an equilibrium for the dynamics $\Phi_G$, where the cost of investment is distributed according to $G$ on $[0,1]$. Let $\bar{G}$ be a CDF that is stochastically dominated by $G$, that is, $\bar{G}(x) > G(x)$ for all $x \in (0,1)$, and both $G$ and $\bar{G}$ are strictly increasing. Then there exists $\bar{\pi} > \pi^*$ such that $\bar{\pi}$ is an equilibrium for $\Phi_{\bar{G}}$.
\end{proposition}

\begin{proof}
	By assumption, we have that $\bar{G}^{-1}(\pi^*) < G^{-1}(\pi^*)$, so $\bar{G}^{-1}(\pi^*) < w\beta(\pi^*)$. Since $\bar{G}^{-1}(1) > \beta(1)$, we must have $\bar{G}^{-1}(\bar{\pi}) = w\beta(\pi)$ for some $\bar{\pi} \in (\pi^*, 1)$.
\end{proof}

% !TeX root = main.tex 

\section{Simulations with non-realizability}
\label{sec:experiments}
In this section we present results from numerical experiments examining the effects of decoupling and subsidies under our model of dynamics, in the absence of group-realizability. Since the type of dynamic data needed for experiments on a real application would require randomized controlled trials, we instead consider a stylized semi-synthetic experiment, based on a widely used FICO credit score dataset from a 2007 Federal Reserve report \citep{fed07}. Importantly, only aggregate statistics were reported and the data we accessed does not contain sensitive or private information. We note that our modeling assumptions may not be realistic for this dataset (see Section~\ref{sec:discussion} for a discussion) and our simulations should not be interpreted as policy recommendations. Instead, these experiments help us illustrate qualitatively the types of dynamics one may find using real world data.

%In this section we present results from numerical experiments on semi-simulated data, . We endeavor 
\paragraph{Stylized Model}
We describe how our model can be instantiated to a \lledit{highly} stylized example of credit scoring and lending. Assume a loan applicant either has the means to repay a loan or not. If they have the means to repay, they always repay ($Y=1$); otherwise they always default ($Y=0$). In order to have the means to repay, applicants must make an \emph{ex ante} investment at the cost of $C$, whose distribution is $\Pr(C < c) = G(c)$. This represents costly actions an individual has to take in order to acquire the financial ability to repay loans, e.g. working at a stable job or taking job preparation classes. Applicants from group $a$ who have the means to repay receive credit scores $X$ drawn from $f_1^a$ and those who don't receive credit scores drawn from $f_0^a$. The decision of the bank is to approve or reject a loan applicant, given their credit scores.

\paragraph{Dataset} FICO scores are widely used in the United States to predict credit worthiness. The dataset, which contains aggregate statistics, is based on a sample of 301,536 TransUnion TransRisk scores from 2003 \citep{fed07} and has been preprocessed by \citet{hardt2016}. These scores, corresponding to $X$ in our model, range from 300 to 850. \lledit{For simplicity, we rescale the scores so that they are between 0 and 1.} Individuals were labeled as defaulted if they failed to pay a debt for at least 90 days on at least one account in the ensuing 18-24 month period. The data is also labeled by race, which is the group attribute $A$ that we use. We  compute empirical conditional feature distributions $\Pr(X=x\mid A=a, Y=y)$ from the available data and fit Beta distributions\footnote{We simulate 100,000 samples from the empirical PDF (see Figure~\ref{fig:distributions}) and fit a Beta distribution by maximum likelihood estimation.} to these to obtain $f_0^a$, $f_1^a$. 

We treat these distributions \emph{as if} they came from our model as shown in Figure~\ref{fig:graph-ourmodel}, for the sole purpose of illustration. This is not to claim that our modeling assumptions hold on this dataset, as discussed earlier. Given the lending domain is complex, our aim is not to faithfully represent this particular domain with our model, but to simulate feature distributions that exhibit group heterogeneity and non-realizability, hence extending our consideration beyond the idealized settings of Sections \ref{sec:realizable} and~\ref{sec:groups}.

 Figure~\ref{fig:distributions} shows the histograms as well as the fitted Beta distributions for $f_0^a$, $f_1^a$, where $a$ is the race attribute.  It is clear that group-realizability does not hold even approximately, since there is significant overlap in the distributions of credit scores for people who repaid and for people who did not repay.

\begin{figure}[t]
\setlength{\abovecaptionskip}{-5pt}
\setlength{\belowcaptionskip}{-2pt}
\centering
	\includegraphics[width=.5\columnwidth, trim=10 0 20 0,clip]{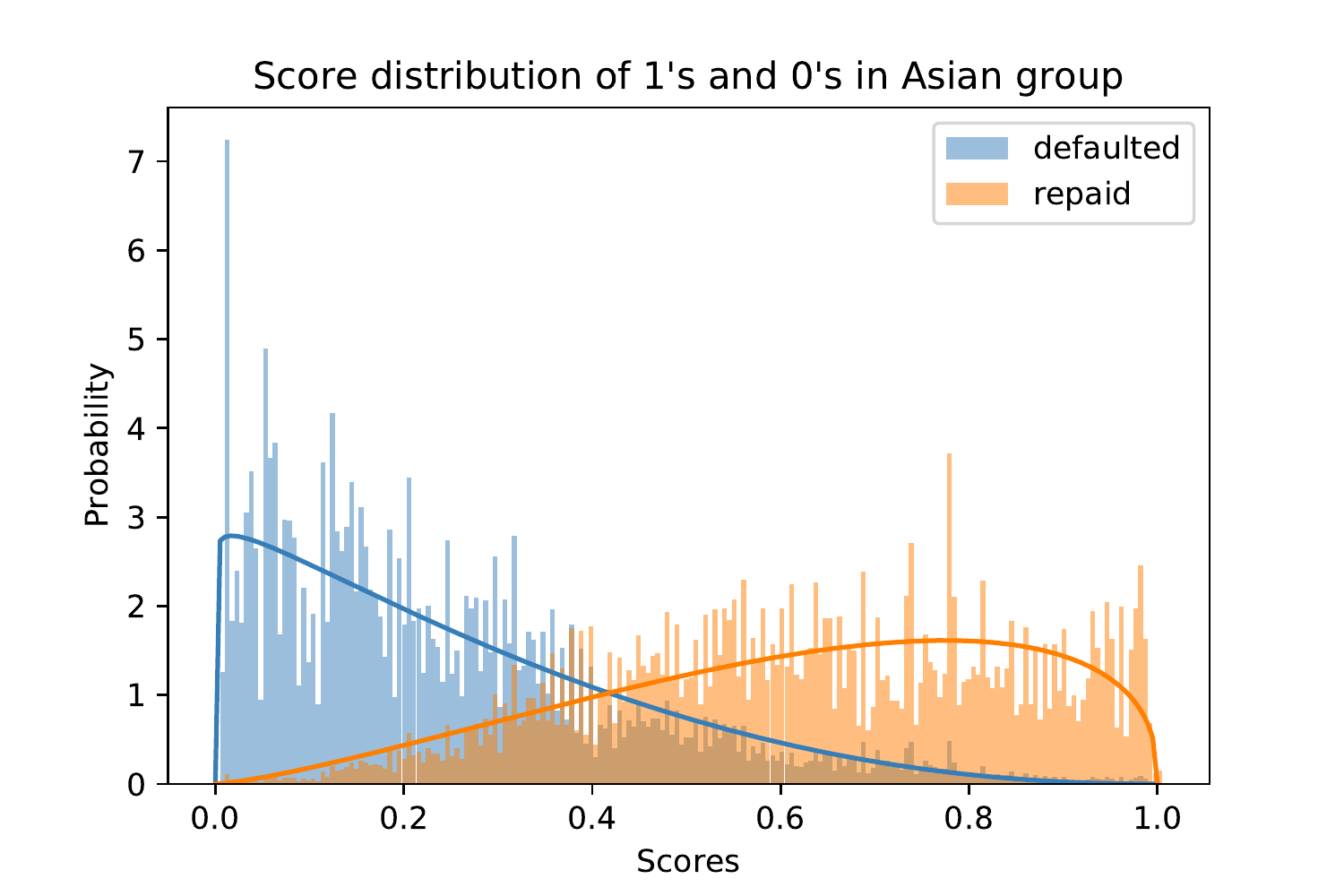}%
	\includegraphics[width=.5\columnwidth, trim=10 0 20 0,clip]{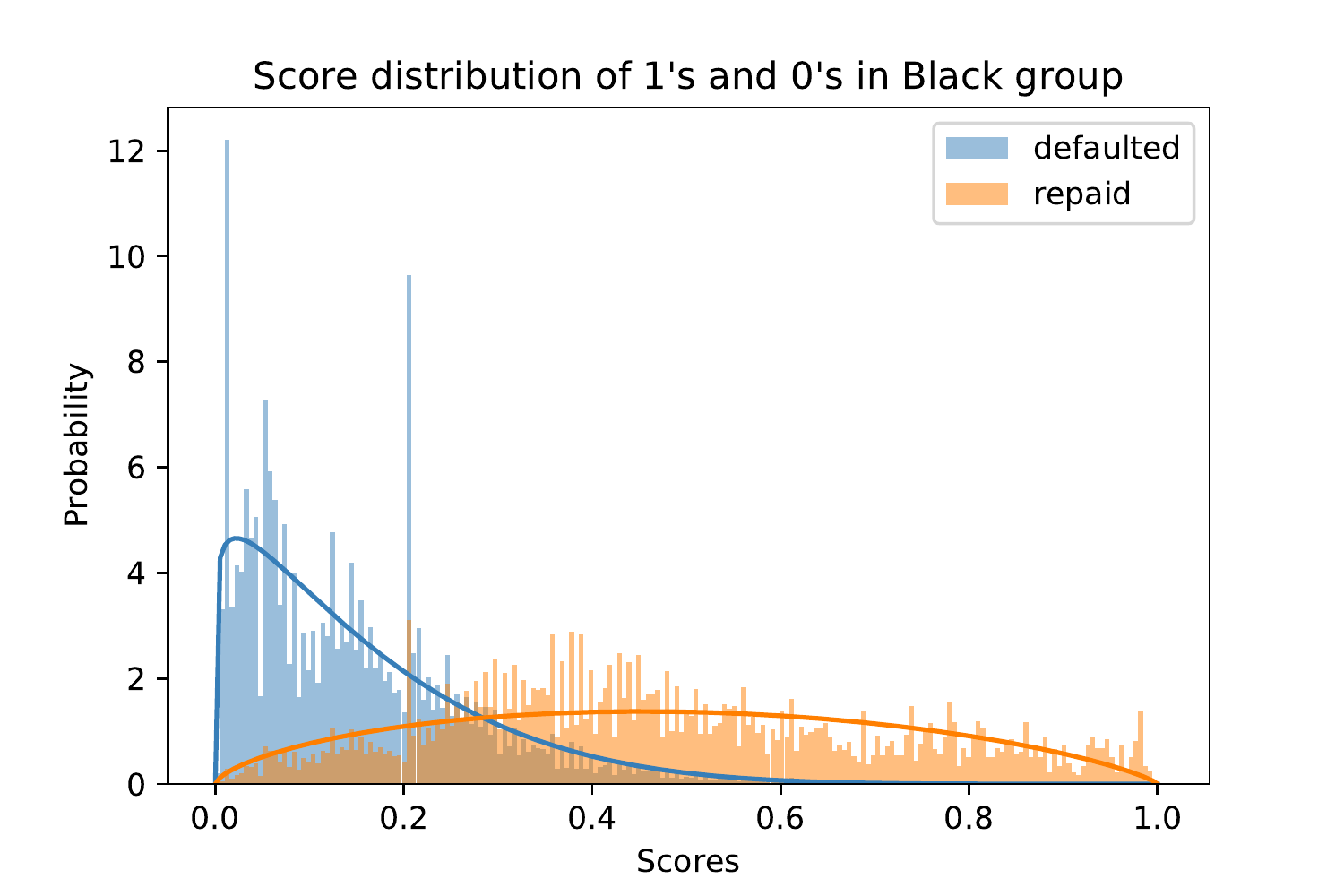}
	\includegraphics[width=.5\columnwidth, trim=10 -10 20 0,clip]{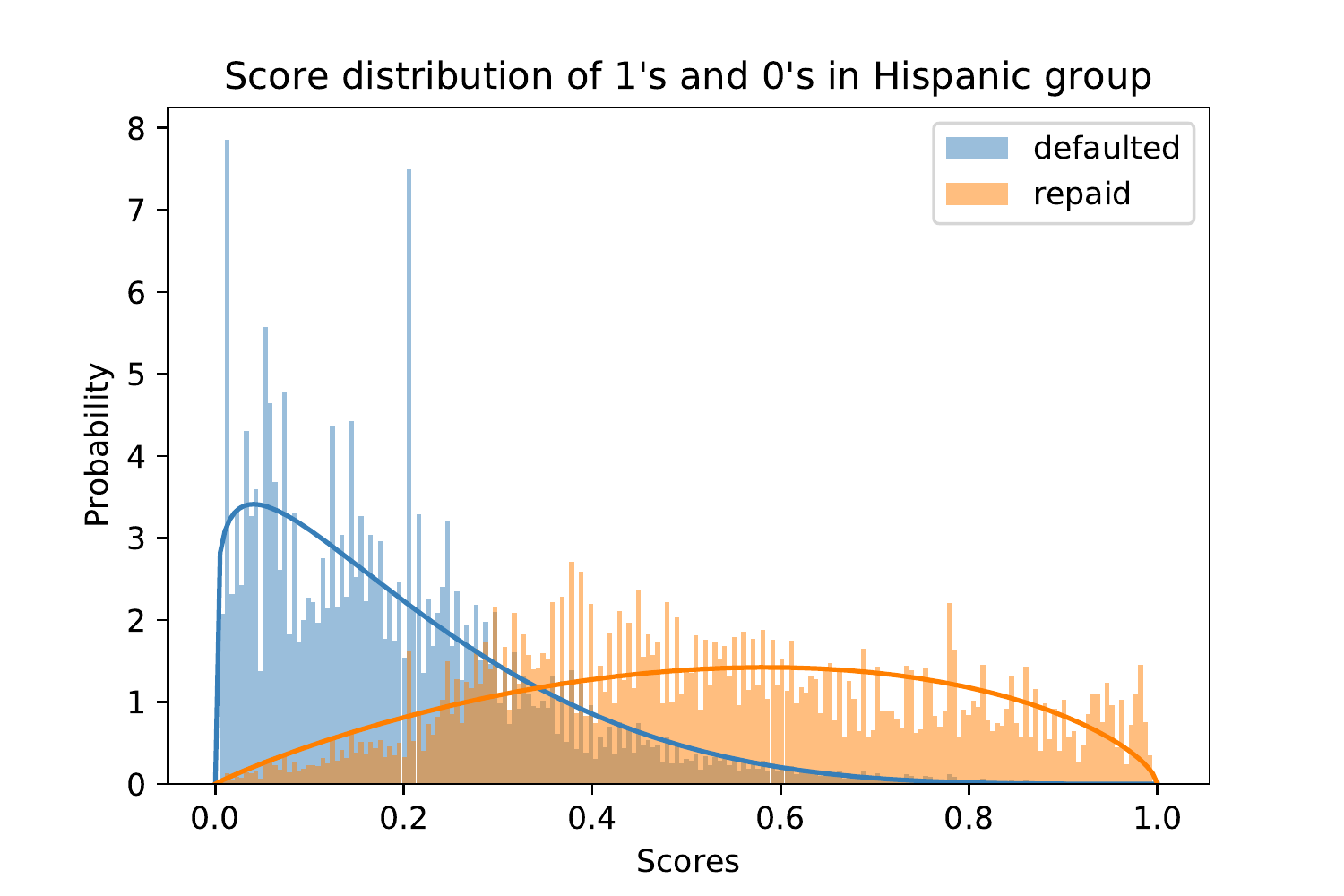}%
	\includegraphics[width=.5\columnwidth, trim=10 -10 20 0,clip]{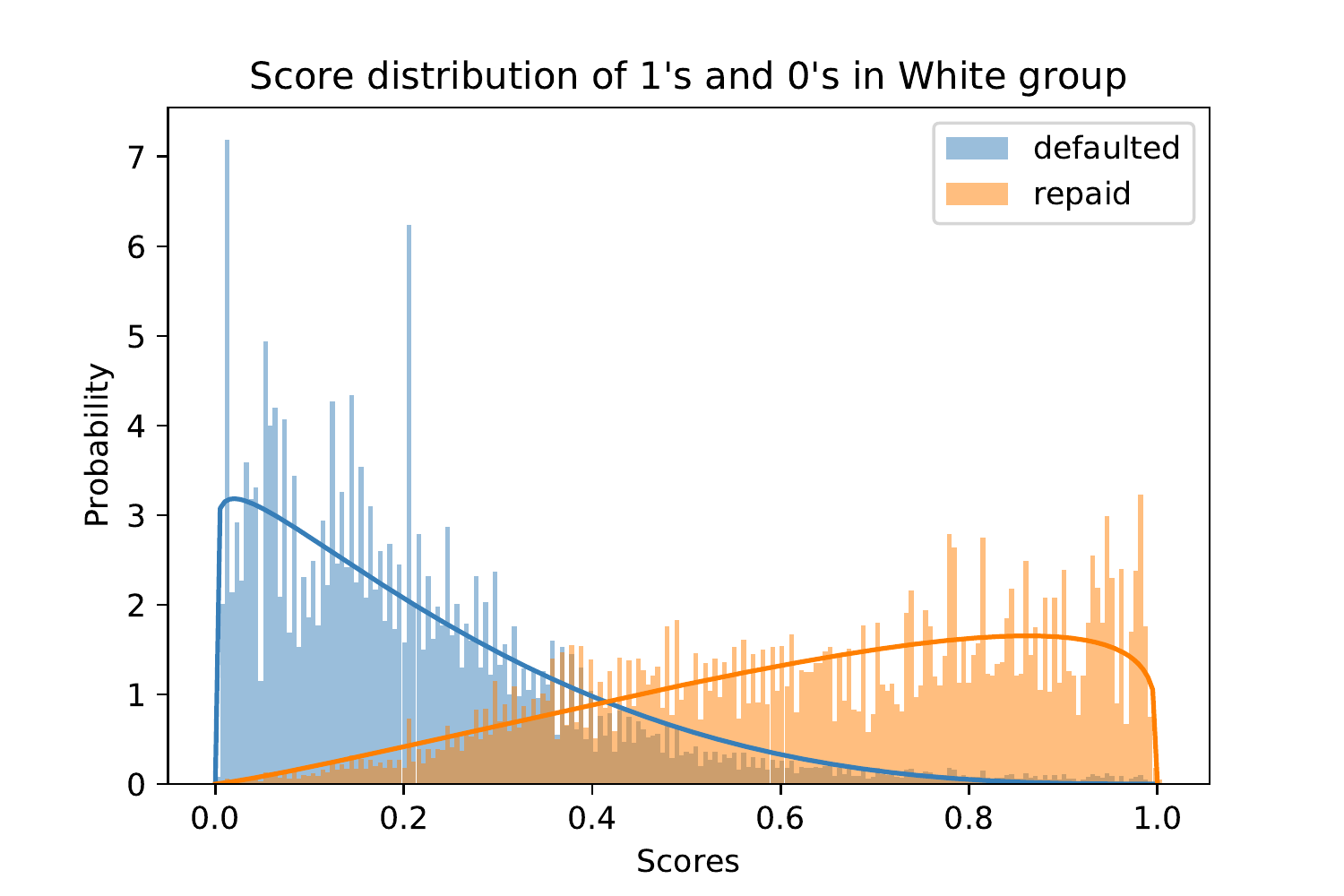}
	\caption{Score distributions conditioning on repayment outcome ($Y$) for different race groups}\label{fig:distributions}
	\vspace{-1em}
\end{figure}

%In the following experiments, we run our model of dynamics forward using different distributions $G$ for the cost of investment, and explore the effects of decoupling and subsidizing.

\subsection{Effects of Decoupling when Multiple Stable Equilibria Exist }

Although decoupling is guaranteed to improve the \investmentlevel~at equilibrium over using a joint decision rule for every group (Sections~\ref{sec:realizable} and \ref{sec:groups}), \lledit{this} is not necessarily true in the non-realizable setting. In fact, even when $G$ is the uniform distribution on $[0,1]$ in all groups \lledit{(i.e. the cost of investment $C$ is uniformly distributed on $[0,1]$, as we considered in Section~\ref{sec:groups})}, decoupling did not benefit all groups. As can be seen from Figure~\ref{fig:decoupling1} in Appendix~\ref{app:expt}, while the White and Asian groups had a higher qualification rate after decoupling, the Black and Hispanic groups saw their \lledit{equilibrium} qualification rate decrease. On the other hand, the effects of decoupling were small in this case (less than 3 percent points difference in the final \investmentlevel). 

We now show that the effect of decoupling can be drastic depending on $G$. Recall that in Section~\ref{sec:non-realizable}, we showed that multiple equilibria, with possibly vastly different \investmentlevels, may exist under the non-realizable setting even when there is a only single group. In general the existence of multiple equilibria depends on properties of $G$, that is, how the cost of investment is distributed in a group. In Figure~\ref{fig:decoupling2}, we show the change in equilibrium investment level after decoupling for an experiment with two groups, Asian and Hispanic. The two plots each correspond to a different bimodal Gaussian distribution for $G$, truncated to $[0,1]$, that have been chosen such that the decoupled dynamics have multiple stable equilibria for the Hispanic (right) and the Asian (left) respectively.

In both plots, we can see that the effects of decoupling depend on the initial \investmentlevel. If the initial \investmentlevel~was too low, or too high, the decoupled dynamics converge to an equilibrium where one of the groups invest \lledit{in qualifications} at a much lower level than they would under the joint dynamics.\lledit{\footnote{Figure~\ref{fig:decoupling-opt1} in Appendix~\ref{app:expt} shows the converged \investmentlevels~of both groups under decoupled and joint dynamics.}}

\begin{figure}[t]
\setlength{\belowcaptionskip}{-10pt}
	\centering
	\includegraphics[width=.5\columnwidth, trim = 0 0 30 0, clip]{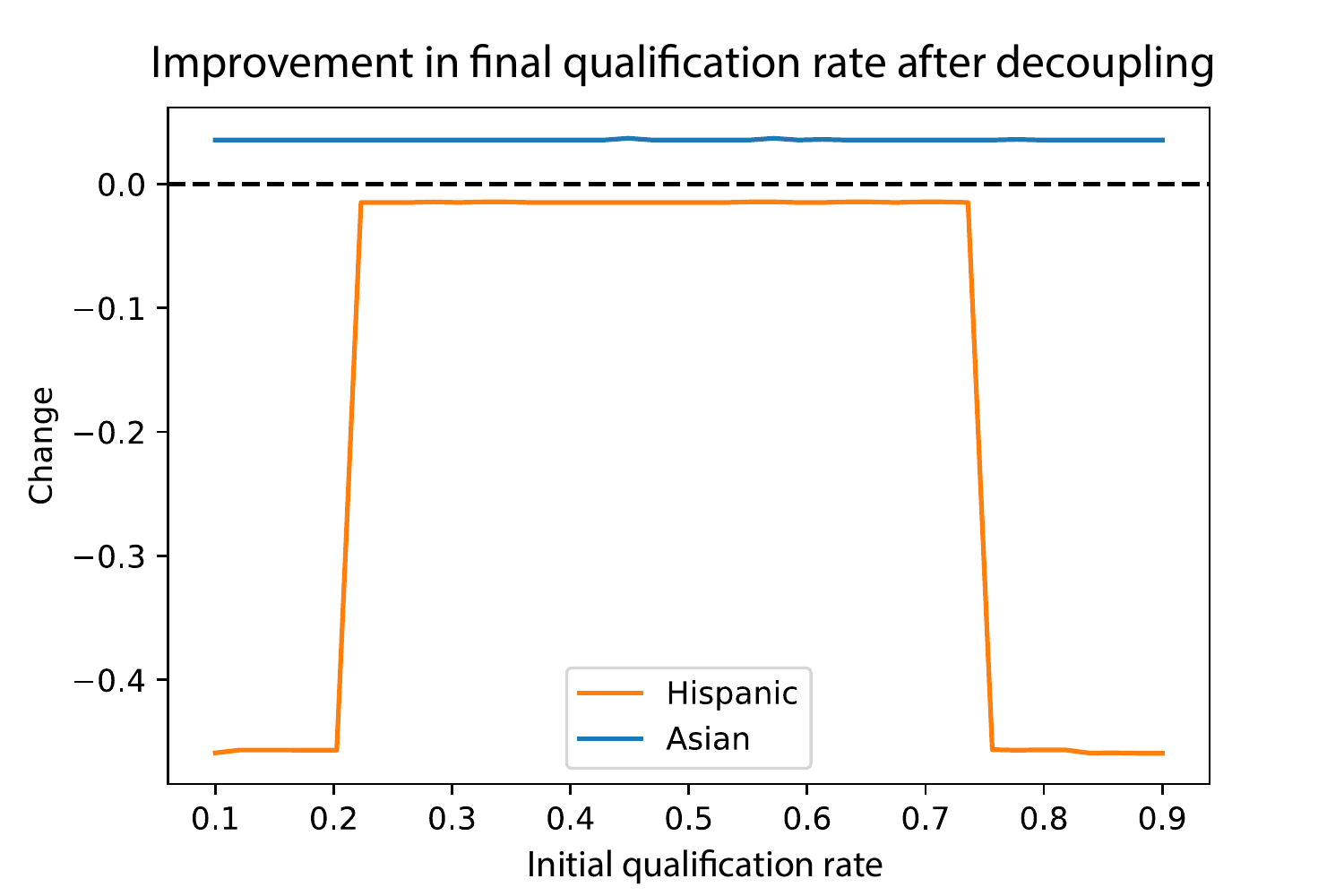}%
	\includegraphics[width=.5\columnwidth, trim = 0 0 30 0, clip]{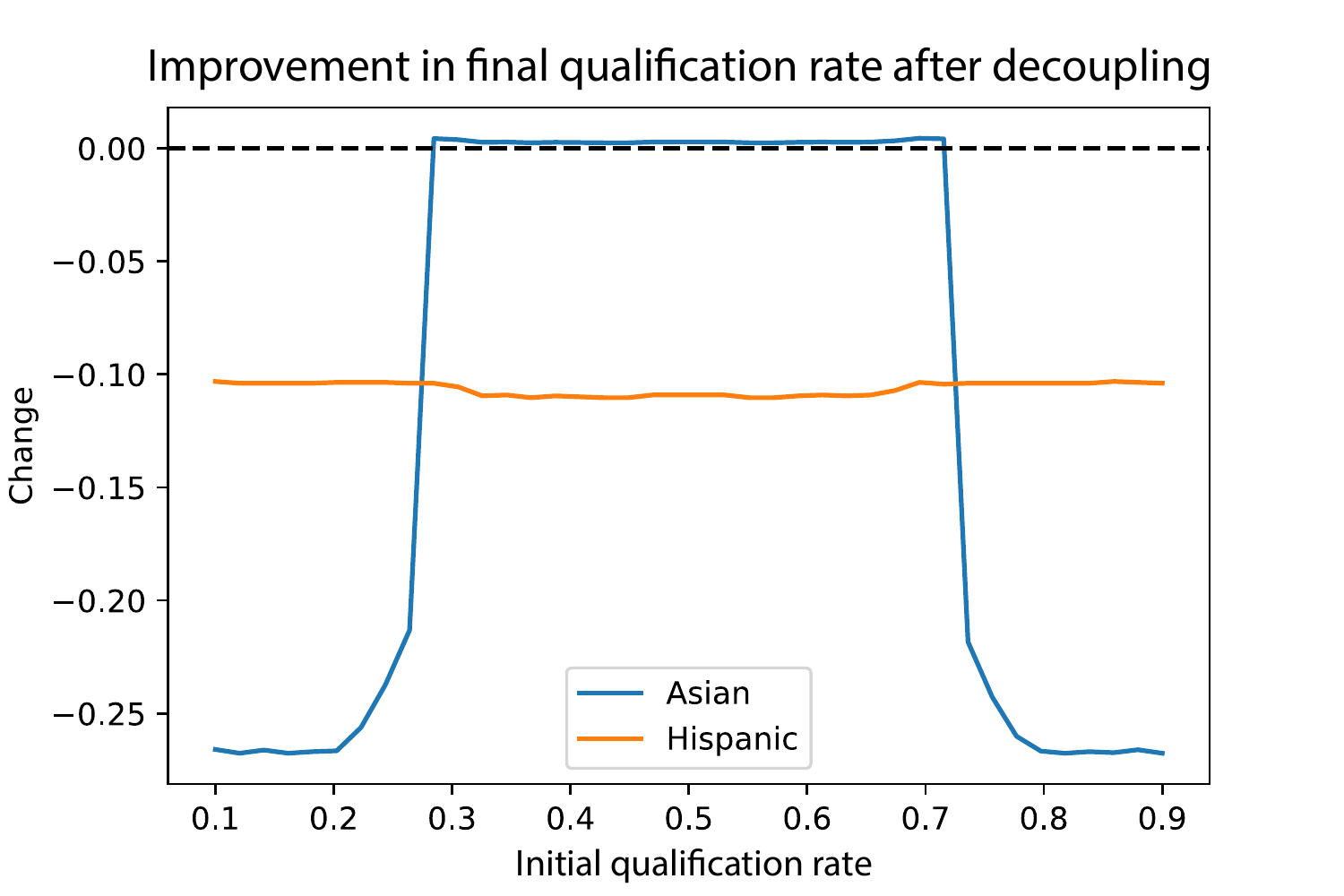}
	\caption{Effects of decoupling in presence of multiple equilibria. We vary the initial qualification rate in the x-axis. The right (resp. left) plot was generated using a bimodal normal distribution for $G$ with modes at 0.57 and 0.74 (resp. at 0.57 and 0.63). }\label{fig:decoupling2}
\end{figure}

\subsection{Subsidizing the Cost of Investment}
In this experiment, we consider if subsidizing the cost of investment of one group by changing $G$ %for the group that is investing the least at equilibrium), does it 
improves their new equilibrium \investmentlevel, under both decoupled and joint dynamics. Specifically, we vary the cost of investment in the Black group.

We use a truncated normal distribution for $G$ and vary its mean (on the x-axis) for a single group, while keeping the other groups' $G$ unchanged (mean of $0.6$).

Figure \ref{fig:subsidy} shows that subsidizing the cost of investment is effective in raising the equilibrium investment level of a group, both in the joint learning and decoupled learning case. Interestingly, large amounts of subsidy for a single group reduced the equilibrium investment levels of other groups. As also suggested by theoretical results in section \ref{sec:subsidies}, subsidizing the \investmentlevel~of one group does sometimes entail a tradeoff in the \investmentlevels~of other possibly more advantaged groups.

Interestingly, lowering the mean cost of investment in the Black group below 0.35 caused the final \investmentlevel~to decrease. This is not a contradiction of Proposition \ref{prop:non-realizable-subsidies}, which argues that \nhedit{equilibria improve under}
%  existence of better equilibria under
 subsidies but does not guarantee that the dynamics will converge to the improved equilibrium. In this case, the decoupled dynamics for the Black group (where the mean cost of investment is $0.30$) actually converged to a limit cycle and the final \investmentlevel~in the plot is an average of the points in the limit cycle. Limit cycles are a challenging object to study in dynamical systems and game theory. While we have commented on their existence in a simple model in group-realizable setting of Section~\ref{sec:multivariate-gaussian}, we leave their implications in the general non-realizable setting to future work.

%\todo{Should we add an earlier note about limit cycles in experiments?}

\begin{figure}[t]
	\centering
	\includegraphics[width=.75\columnwidth, trim = 0 0 10 0, clip]{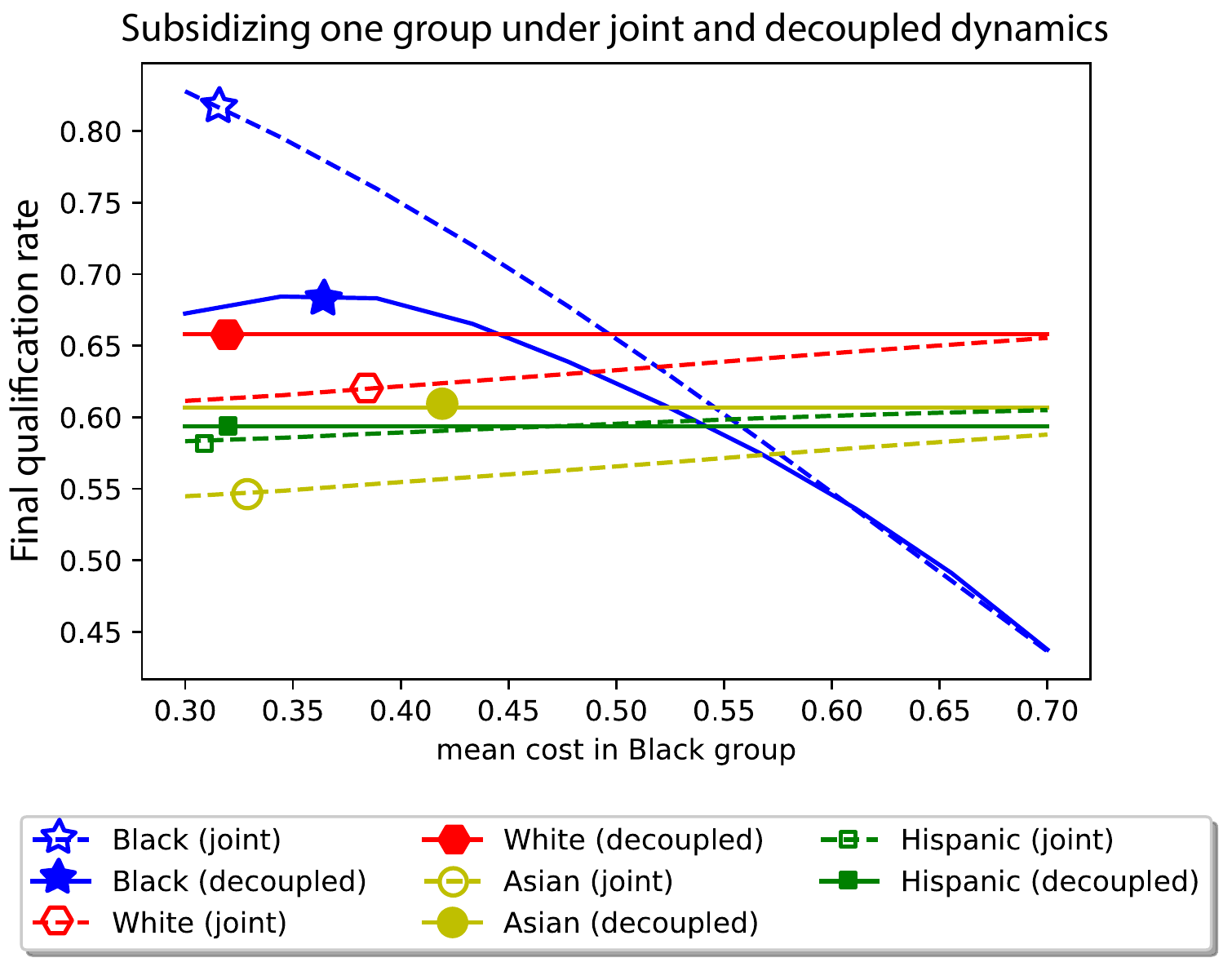}%

	\caption{Effects of raising the average cost of investment, by varying the mean of $G$ on the $x$-axis.}\label{fig:subsidy}
\end{figure}

% !TeX root = main.tex 

\section{Related Work}\label{sec:related}

Our work follows a growing line of work on how machine learning algorithms interact with human actors in a dynamic setting, with the goal of understanding and mitigating disparate impact.

Recent work examine the long-term impact of group fairness criteria \citep[see e.g.,][Chapter 2]{barocas-hardt-narayanan} on automated decision making systems: \citet{liu18delayed} show that static fairness criteria fail to account for the delayed impact that decisions have on the welfare of disadvantaged groups. In the context of hiring, however, \citet{Hu2018shortterm} find that applying the demographic parity constraint in a temporary labor market achieves an equitable long-term equilibrium in the permanent labor market by raising worker reputations.

Prior work on the fairness of machine learning has also examined tradeoffs between fairness criteria \citep{kleinberg16inherent, chould16fair}, as well as the incompatibility between risk minimization and fairness criteria \citep{liu18implicit}, assuming that the qualification rates differ across groups. These results concern the static setting, whereas we highlight the fact that qualification rates tend to change in response to the decision rules in place. 

Another line of work \citep{hashimoto18demographics,zhang2019longterm} analyzes a dynamic model where users respond to errors made by an institution by leaving the user base uniformly at random, and demonstrate how the risk-minimizing approach to machine learning can amplify representation disparity over time. This is complementary to our work which models individuals as \lledit{rational %active
decision makers who may or may not have the incentive to acquire the positive label.} In particular, \citet{hashimoto18demographics} show that equilibria with equal user representation from all groups can be unstable, and that robust learning can stabilize such equilibria. Unlike in our model, \lledit{the user representation model does not distinguish between positive and negative labels, and thus do not distinguish between false negative and false positive errors}. This is a crucial distinction in high-stakes decision making as different error types present asymmetric incentives \lledit{for individuals}, as explained in Section~\ref{sec:model}; for example, a high false positive rate in hiring would encourage under-qualified job applicants.

 %\citep{Hu2018shortterm,liu18delayed, hashimoto18demographics, Mouzannar2019socialequality, Kilbertus19Fair}. 
%Yet the influence of the resulting decision systems on society has only garnered
%modest treatment 
%
%
 %~\cite{ }, 
%
%\todo{distinguish model with \citep{hashimoto18demographics}. different domains, decision making vs. serving predictions.} 

\citet{Hu2019disparate} and \citet{Milli2019social} study the disparate impact of being robust towards strategic manipulation \citep[see e.g.,][]{Hardt:2016:SC}, where individuals respond to machine learning systems by manipulating their features to get a better classification. In contrast to our model (Figure~\ref{fig:subsidy}), their setting models the individual as intervening directly on their features, $X$, and this is assumed to have no effect on their qualification $Y$. This assumption applies to features that are easy to `game' (e.g. scores on standardized tests can be improved by test preparation classes) but is less applicable to features that more directly correspond to \emph{investment} in one's qualifications (e.g. taking AP courses in high school). \citet{Hu2019disparate} also show that subsidizing the costs of the disadvantaged group to strategically manipulate their features can sometimes lead to harmful effects. 
%
%looked at the fairness implications of being robust towards strategic manipulation of features. 
%
%\todo{explain the difference with these models. for our model, emphasize the group aspect of the dynamics, importance of investment.} 
%
\citet{kleinberg18investeffort} study a non-dynamic setting where some features of the individual result in higher qualification and some do not, as described by an ``effort profile'', which is a direct manipulation of $X$, that the institution wants to incentivize the individual to pursue.

Our work is also related to the topic of statistical discrimination in economics \citep{phelps72statistical, Arrow73theory,Arrow98economics}, which studies how disparate market outcomes \lledit{at} equilibrium can arise from imperfect information. This line of work often involves wage discrimination, whereas we assume the wage is fixed and standard for all groups. \citet{coate93will} proposed a model of rational individual investment in the labor market under a fixed wage and showed that affirmative action\footnote{In this work, affirmative action is defined as constraining the rate of acceptance, $\Pr(\hat{Y}_\theta = 1\mid A=a)$, to be equal in both groups.} may lead to an undesirable equilibrium where one group still invests sub-optimally. The model in our work is most closely related to their model, with two key distinctions: 1)~We allow features $X$ to be multi-dimensional, whereas~\citet{coate93will} assumes that $X$ is a one-dimensional `noisy signal'. 
2)~We consider the case where the conditional feature distributions,  $\Pr(X=x \mid Y=y, A=a)$, differ by groups whereas~\citet{coate93will} assumes that the groups are identically distributed. Note that under our models, if the conditional feature distributions were shared across groups, then \emph{any} hiring policy will result in fully balanced equilibria where all groups have the same \investmentlevel~and are hired at the same rate. This does not corroborate with reality, where conditional feature distributions do in fact differ across groups and we routinely observe institutions applying the same model to all individuals only to see obviously discriminatory outcomes \citep{amazon}. By modeling feature heterogeneity across groups, we find that it necessarily leads to disparate equilibria.

%\citet{blume06dynamics} studied dynamics induced by a simple model that does not
Recently, \citet{Mouzannar2019socialequality} studied the equilibria of \investmentlevels~under a generic class of dynamics, focusing on contractive maps and the effects of affirmative action. In contrast, our work motivates a model of dynamics based on rational investment, and this typically leads to non-contractive dynamics. We are both interested in balanced equilibria, which \citet{Mouzannar2019socialequality} terms `social equality'.

Finally, our work studies two interventions for finding more desirable equilibria: decoupling the classifier and subsidizing the cost of investment.  
%decoupling classifiers has not been as well 
Several works, including 
\citet{dwork18a} and \citet{pmlr-v97-ustun19a} have studied decoupled classifiers in the static classification setting. Our work sheds light on  when such interventions are useful in the dynamic decision making setting.  %\citet{corbett18measure} ?
%\todo{}

%\todo{Cite Solon Barocas etc. re: promise of machine learning is to increase accuracy, and how that requires realizability. }
%``data mining also has the potential to help reduce discrimination by forcing decisions onto a more reliable and empirical foundation''~\cite[p. 676]{Barocas16,Zarsky} \citet{Passi:2019:PFF}?

%\citet{Chouldechova2018TheFO}

% !TeX root = main.tex 

\section{Discussion and Future Work}\label{sec:discussion}
In this work, we have made the following contributions:
\begin{enumerate}[leftmargin=*]
	\item We proposed a dynamic model of automated decision making where individuals invest rationally based on the current assessment rule. We studied the properties of equilibria under these dynamics.
	\item We showed that common properties of real data, namely heterogeneity across groups and the lack of realizability, lead to undesirable tradeoffs at equilibria, resulting in long term outcomes that disadvantage one or more groups.
	\item We considered two interventions---decoupling and subsidizing the cost of investment---and showed that they have a significant impact on the nature of equilibria both in theory and in numerical experiments.
\end{enumerate}

We now discuss the limitations of the current work and avenues for future research. Questions related to finite sampling and its ramifications for the nature of equilibria are challenging and warrant further study. This work assumed that the institution has unbiased (and zero-variance) estimates of true and false positive rates over the entire population, even though it presumably can only observe the qualification of candidates \emph{after} hiring them. This is known as the selective labeling problem, which could introduce bias. In theory, unbiased estimates can be achieved by a small degree of random sampling and appropriate reweighting \citep[see e.g.,][]{Kilbertus19improving}, but this is still a large problem in practice that requires domain-specific knowledge and solutions~\citep{Arteaga18selectivelabels,Kallus2018Residual}.

Our model assumed that individuals make a rational decision to invest and can affect their qualification $Y$ directly. This assumption could be reasonable in settings like hiring, for example, where investing to acquire skills usually leads to increased competence. In some settings, however, individuals may be unable to effectively intervene on $Y$. For example, a business loan applicant who \lledit{is a good business operator} could still default on their loan due to external economic shocks or other forms of disadvantage that have not been taken into account. In this case, the current model does not fully capture the complex societal processes that lead to a positive outcome. \lledit{Our work nonetheless shows that even in an idealized setting where individuals can effectively and rationally intervene on their outcome labels $Y$, underlying factors such as heterogeneity across groups and non-realizability already lead to undesirable tradeoffs at equilibrium. We leave the extensions of the current model beyond rational individual investment to future work.}
%and much more work is required to understand the dynamics in these domains.

\paragraph{Acknowledgements} The authors thank Moritz Hardt, Katherine van Koevering, and Ellen Vitercik for valuable feedback.

\bibliographystyle{abbrvnat}

\bibliography{mybib}
\pagebreak 

\appendix

	% !TeX root = main.tex 

\section{Proof of Theorem~\ref{thm:eqi_approx_real}}\label{app:theoremproof}
%\begin{proof}
For any $\pi \in  [s,1-s]$, consider the profit-maximizing classifier,
\[\brtheta(\pi) = \argmax_{\theta\in \Theta} ~\text{TPR}(\theta) \cdot \pi - \text{FPR}(\theta) \cdot (1-\pi). \]

Let $\pi(0)$ be the initial qualification rate. For ease of notation, denote $\theta^* = \brtheta(\pi(0))$. We examine the new qualification rate $\pi_1$ under the best response model~$\theta^*$. Since there exists a $\theta$ such that $ \tpr(\theta) \ge 1- \epsilon$ and $\fpr(\theta) \le \epsilon$, we have $\text{TPR}(\theta) \cdot \pi - \text{FPR}(\theta) \cdot (1-\pi) \geq \pi-\epsilon$. It follows that
\[\text{TPR}(\theta^*) \cdot \pi - \text{FPR}(\theta^*) \cdot (1-\pi) = \pi(\tpr(\theta^*) - \fpr(\theta^*)) + (2\pi-1)\fpr(\theta^*) \geq \pi-\epsilon\]

Rearranging this inequality gives the following lower bound on $\tpr(\theta^*) - \fpr(\theta^*)$:
\begin{align}
	\tpr(\theta^*) - \fpr(\theta^*) \geq \frac{\pi-\epsilon -(2\pi-1)\fpr(\theta^*)}{\pi} \label{eq:1}
\end{align}
For $\pi < 1/2$, it follows from \eqref{eq:1} that
\begin{align}
	\tpr(\theta^*) - \fpr(\theta^*) 
	&\ge \frac{\pi-\epsilon}{\pi} \tag{using $\fpr(\theta^*) \ge 0$} \\
	&\geq 1 - \frac{\epsilon}{s}
\end{align}

For $\pi > 1/2$, we have
\begin{align*}
	\fpr(\theta^*) \leq \frac{\pi\tpr(\theta^*) - \pi + \epsilon}{1-\pi} \leq \frac{\epsilon}{1-\pi}.
\end{align*}

Substituting this into \eqref{eq:1} gives:
\begin{align}\label{eq:useincor}
	\tpr(\theta^*) - \fpr(\theta^*) &\ge  \frac{\pi-\epsilon -(2\pi-1)\frac{\epsilon}{1-\pi}}{\pi} = 1-\frac{\epsilon}{1-\pi} \geq 1- \frac{\epsilon}{s}
\end{align}

Therefore the new qualification rate $\pi_1$ satisfies $\pi_1 > G(w(1-\epsilon/s))$.

Notice that $\pi_1 \le G(w) \le 1-s$ and $ \pi_1 > G(w(1-\epsilon/s)) \geq G(w) -  L_Gw\epsilon/s \ge s$, so \lledit{we may repeat the argument to conclude that the qualification rate in the limit must be greater than  $G(w(1-\epsilon/s))$.}

\section{Proof of Corollary~\ref{cor:subsidy}}\label{app:corproof}
From Theorem~\ref{thm:eqi_approx_real} applied to investment level $\bar G$, we can conclude $\bar \pi \geq \bar G(w(1 - \epsilon/s))$. By assumption $\bar G(w(1 - \epsilon/s)) \geq G(w)$, thus it remains to show $G(w) \geq \pi$. However, this follows immediately from the fact that $1 \geq \tpr(\theta) - \fpr(\theta)$, $\forall \theta\in \Theta$.
%\end{proof}
	% !TeX root = appendix_supp.tex 
\section{Supplementary material and proofs for Section \ref{sec:uniform}}
\label{app:proof-uniform}
\begin{proof}[Proof of Proposition \ref{prop:uniform-eq} ]

First consider the policy $\hat{Y}_1 = \Ind{X > h_1}$. Given this policy, we have $\pi_{a_1} = G(w)$ and $\pi_{a_2} = G(w(1-\frac{h_2-h_1}{h_2})) = G\left(w \cdot \frac{h_1}{h_2}\right) $. $\hat{Y}_1$ is optimal for this $\pi$ if the gain from true positives in group $a$ offsets the loss from false positives in group $a_2$ for $X \in [h_1, h_2]$, i.e. we need
	\begin{equation}\label{eq:h1}
	\frac{G(w)\cdot n_{a_1} \cdot p_{\text{TP}}}{1-h_1}  > 		\frac{\left(1-G\left(w \cdot \frac{h_1}{h_2}\right)\right)\cdot n_{a_2} \cdot c_{\text{FP}}}{h_2}.
	\end{equation}
	Now consider the policy $\hat{Y}_2 = \Ind{X > h_2}$. Given this policy, we have $\pi_{a_1} = G\left(w\cdot\frac{1-h_2}{1-h_1}\right)$ and $\pi_{a_2} = G(w)$. $\hat{Y}_2$ is optimal for this $\pi$ if the gain from true positives in group $a$ fails to offset the loss from false positives in group $a_2$, for $X \in [h_1, h_2]$, i.e. we need
	\begin{equation}\label{eq:h2}
	\frac{G\left(w\cdot\frac{1-h_2}{1-h_1}\right)\cdot n_a \cdot p_{\text{TP}}}{1-h_1} <		\frac{\left(1-G\left(w \right)\right)\cdot n_{a_2} \cdot c_{\text{FP}}}{h_2}.
	\end{equation}
	Direct computation shows that \eqref{eq:h1} and \eqref{eq:h2} are satisfied as long as $w$ lies in the interval \[\left( \frac{h_2(1-h_1)}{h_2^2+h_1(1-h_1)},\frac{(1-h_1)^2}{(1-h_2)h_2+(1-h_1)^2} \right).\]
	
	Both equilibria above are stable since \eqref{eq:h1} and \eqref{eq:h2} hold with strict inequality. For all small enough perturbations to $(\pi_{a_1}, \pi_{a_2})$, $h_1$ (or $h_2$) will still remain as the profit maximizing threshold.
	
	There exists an equilibrium at $h=h_1+g$ if we have \begin{equation}
	\frac{G(w\cdot \frac{1-h_1-g}{1-h_1})}{1-h_1} = \frac{1-G(w\cdot \frac{h_1+g}{h_2})}{h_2} 
	\end{equation}
	Direct computation shows that the above equation is satisfied by a unique value of \[g=\frac{(1-h_1)(-wh_2^2+h_2(1-h_1)-wh_1(1-h_1))}{w((1-h_1)^2-h_2^2))},\] and that $\pi_{a_1} = \frac{w(1-h_1-g)}{1-h_1}= \pi_{a_2} =\frac{w(h_1+g)}{h_2}$ if $w=1-h_1$.
	
	\end{proof}
For an illustration of Proposition~\ref{prop:uniform-eq}, consider an example where $n_{a_1} \cdot p_{\text{TP} } = n_{a_2} \cdot c_{\text{FP}}$, $G(c) = c$ for $c < 1$ (uniformly distributed cost of investment), and $h_1 = 0.4$, $h_2 = 0.8$, which gives $h_1/h_2 = 0.5$. Let $w = 0.6$. We compute:
\[ \frac{G(w)}{1-h_1} = 1, \frac{1-G\left(w \cdot \frac{h_1}{h_2}\right)}{h_2}=0.875, \frac{G\left(w\cdot\frac{1-h_2}{1-h_1}\right)}{1-h_1} =  1/3,\frac{ 1-G(w)}{h_2} = 1/2, \] and check that these satisfy the assumptions.

In this example, note that the first equilibrium has \investmentlevel~$0.6$ and $0.3$ for groups $a$ and $a_2$ respectively, while the second equilibrium has \investmentlevel~$0.2$ and $0.6$ respectively. The first equilibrium might be more desirable since there is higher \investmentlevel~overall, though neither equilibrium has equal \investmentlevels~across the two groups.

\paragraph{Comparison of equilibria}  We can compare the equilibria described in Proposition~\ref{prop:uniform-eq} in terms of metrics shown in Table \ref{tab:uniform}. There is no fixed ranking for the Institution's utility; instead the ranking varies depending on the values of $h_1, h_2$, and $w$.

\begin{table}
	\begin{center}
	\scalebox{0.85}{
		\begin{tabular}{| p{3cm} |  p{4cm}|   p{2.5cm}|  p{4.5cm}| p{3cm}|} 
			\hline
			\thead{\bf Equilibrium $h$} & \thead{$h_1$ }&\thead{ $h_2$} & \thead{$h_m := h_1+g$} & \thead{\bf Ranking} \\ [0.5ex] 
			\hline
			\hline
			\makecell{Stability}  & \makecell{Stable} & \makecell{Stable}  &  \makecell{Unstable} & \makecell{-} \\ 
			\hline
			Qualification rate in group $a_1$, $\pi_{a_1}$  & \makecell{$w$} &\makecell{$w\cdot\frac{1-h_2}{1-h_1}$}  & \makecell{$\frac{(w-h_2)(1-h_1)}{(1-h_1)^2-h^2_2}$} & $h_1 \succ h_m \succ h_2$ (Lem.~\ref{lem:skilla1}) \\ 
			\hline
			Qualification rate in group $a_2$, $\pi_{a_2}$   &  \makecell{$ w\cdot \frac{h_1}{h_2}$} &\makecell{ $w$} & \makecell{$\frac{(1-h_1)^2-wh_2}{(1-h_1)^2-h_2^2}$} & $h_2 \succ h_m\succ h_1$ (Lem.~\ref{lem:skilla2})  \\ 
			\hline
			Balance in qualification rate $|\pi_{a_1}-\pi_{a_2}|$ & \makecell{$w\cdot \frac{h_2-h_1}{h_2}$} & \makecell{$w\cdot \frac{h_2-h_1}{1-h_1}$} & \makecell{$\frac{|1-h_1-w|}{h_2-(1-h_1)}$} & $h_m\succ h_1 \succ h_2$ (Lem.~\ref{lem:balance}) \\ 
			\hline
%			Total income received by group $a_1$ & \makecell{$w^2$ } & \makecell{$\left(w\cdot\frac{1-h_2}{1-h_1}\right)^2$} &\makecell{$\left(\frac{(w-h_2)(1-h_1)}{(1-h_1)^2-h^2_2}\right)^2$ }& \makecell{$h_1 \succ h_m \succ h_2$}\\
%			\hline
%			Total income received by group $a_2$ &\makecell{$\left(w\cdot\frac{h_1}{h_2}\right)^2 +w \left(1 -  \frac{h_1}{h_2} \right)$} &\makecell{ $w^2$} &$\makecell{\frac{\left((1-h_1)^2-wh_2\right)h_2(h_2-w)}{\left((1-h_1)^2-h_2^2\right)^2
%			} + w}$ &    \makecell{no ranking} \\%$h_2 \succ h_1$ for $w > \frac{h_2}{h_1+h_2}$
			
			%?? \\
			%\hline
			%Firm's profit & $w+w\cdot \frac{h_1}{h_2}-(1-w\cdot \frac{h_1}{h_2})\cdot \frac{h_2-h_1}{h_2}$ & $w\cdot\left(\frac{1-h_2}{1-h_1}\right)^2 + w$  & $\frac{1}{w}\left(\frac{(w-h_2)(1-h_1)}{(1-h_1)^2-h^2_2}\right)^2+\frac{(1-h_1)^2-wh_2}{(1-h_1)^2-h^2_2} - (1- \frac{(1-h_1)^2-wh_2}{(1-h_1)^2-h_2^2})(1- \frac{(1-h_1)^2-wh_2}{w((1-h_1)^2-h_2^2)})$  & ?? \\ [.5ex] 
			\hline
		\end{tabular}}
		\caption{Comparison of equilibria for uniform scores. In this table we refer to each equlibria using the associated threshold decision policy.}\label{tab:uniform}
	\end{center}
\end{table}

\begin{lemma}[Skill acquisition in group $a_1$]\label{lem:skilla1}
	Let $w \in (w_l, w_u)$, as defined in \eqref{eq:w-bounds}. Then for $h_1, h_2$, as defined in Proposition \ref{prop:uniform-eq}, we have 
	\begin{equation}
	w\cdot\frac{1-h_2}{1-h_1} < \frac{(w-h_2)(1-h_1)}{(1-h_1)^2-h^2_2} < w
	\end{equation}
\end{lemma}
\begin{proof}
	First we show that $\frac{(w-h_2)(1-h_1)}{(1-h_1)^2-h^2_2} < w$ for all $w \in (w_l,w_u)$. It suffices to show that \[ \left(1-\frac{h_2}{w}\right)(1-h_1) \le(1-h_1)^2-h^2_2  \] holds for $w = w_u$, since the LHS is strictly increasing in $w$. We may check by computation that the above in fact holds with equality.
	
	Next we show that $\frac{(w-h_2)(1-h_1)}{(1-h_1)^2-h^2_2} >	w\cdot\frac{1-h_2}{1-h_1} $ for all $w \in (w_l,w_u)$. This amounts to showing that \[\left(1-\frac{h_2}{w}\right)(1-h_1)^2 \ge(1-h_2)((1-h_1)^2-h^2_2) ),  \] for $w = w_l$, since the LHS is strictly increasing in $w$.  We may check by computation that the above in fact holds with equality.
\end{proof}

\begin{lemma}[Qualification rate in group $a_2$]\label{lem:skilla2}
	Let $w \in (w_l, w_u)$, as defined in \eqref{eq:w-bounds}. Then for $h_1, h_2$, as defined in Proposition \ref{prop:uniform-eq}, we have 
	\begin{equation}
	 w\cdot \frac{h_1}{h_2} < \frac{(1-h_1)^2-wh_2}{(1-h_1)^2-h_2^2} < w
	\end{equation}
\end{lemma}
\begin{proof}
	First we show that $\frac{(1-h_1)^2-wh_2}{(1-h_1)^2-h_2^2} < w$ for all $w \in (w_l,w_u)$. It suffices to show that \[ \frac{(1-h_1)^2}{w} - h_2 \le(1-h_1)^2-h^2_2 ) \] holds for $w = w_l$, since the LHS is strictly decreasing in $w$. We may check by computation that the above in fact holds with equality.
	
	Next we show that $ w\cdot \frac{h_1}{h_2} < \frac{(1-h_1)^2-wh_2}{(1-h_1)^2-h_2^2} $ for all $w \in (w_l,w_u)$. This amounts to showing that \[\frac{h_2(1-h_1)^2}{w} - h_2^2 \ge h_1((1-h_1)^2-h^2_2 )  \] for $w = w_u$, since the LHS is strictly decreasing in $w$.  We may check by computation that the above in fact holds with equality.
\end{proof}

\begin{lemma}[Unstable Equilibrium is the most balanced]\label{lem:balance}
	Let $w \in (w_l, w_u)$, as defined in \eqref{eq:w-bounds}. Then for $h_1, h_2$, as defined in Proposition \ref{prop:uniform-eq}, we have 
	\begin{equation}
	\frac{|1-h_1-w|}{h_2-(1-h_1)} < w\cdot \frac{h_2-h_1}{h_2} < w\cdot \frac{h_2-h_1}{1-h_1}
	\end{equation}
\end{lemma}
\begin{proof}
First note that since $h_2 > 1-h_1$ by assumption, we have $ w\cdot \frac{h_2-h_1}{h_2} < w\cdot \frac{h_2-h_1}{1-h_1}$, so it suffices to show that \[	\frac{|\frac{1-h_1}{w}-1|}{h_2-(1-h_1)} <  \frac{h_2-h_1}{h_2}\] for all $w \in (w_l, w_u)$.
Consider the first case: $w \in (w_l, 1-h_1]$. We want to show \begin{equation}\label{eq:lem1-1}
\frac{\frac{1-h_1}{w}-1}{h_2-(1-h_1)} <  \frac{h_2-h_1}{h_2}.
\end{equation} Since the LHS is decreasing in $w$, it suffices to show that \eqref{eq:lem1-1} holds for $w = w_l$. By computation, we have following:
\begin{equation}
\frac{\frac{1-h_1}{w_l}-1}{h_2-(1-h_1)} <  \frac{h_2-h_1}{h_2} \iff (h_2+1)(h_2-h_1)(h_1+h_2-1) > 0,
\end{equation}
which is indeed satisfied since we have $h_2 > h_1$ and $h_2 > 1-h_1$ by assumption.

Now consider the second case: $w \in (w_l, 1-h_1)$. We want to show \begin{equation}\label{eq:lem1-2}
\frac{1-\frac{1-h_1}{w}}{h_2-(1-h_1)} <  \frac{h_2-h_1}{h_2}. 
\end{equation} 
Since the LHS is increasing in $w$, it suffices to show that \eqref{eq:lem1-2} holds for $w= w_u$. By computation, we have following:
\begin{equation}
\frac{1-\frac{1-h_1}{w_u}}{h_2-(1-h_1)} <  \frac{h_2-h_1}{h_2} <  \frac{h_2-h_1}{h_2} \iff (h_2-h_1)(h_1+h_2-1) > 0,
\end{equation}
which is indeed satisfied since we have $h_2 > h_1$ and $h_2 > 1-h_1$ by assumption.

\end{proof}
	% !TeX root = main.tex 
\section{Supplementary material and proofs for Section \ref{sec:multivariate-gaussian}
 }\label{app:proof-mvgaussians}

\begin{proof}[Proof of Proposition \ref{prop:gaussian-eq} ]
	Denote $r:=\frac{\ptp}{\cfp}$. For any hyperplane $h$, we may compute the true positive rate and false positive rate for a group with hyperplane $h_i$ as follows:
	\begin{equation}
	\tpr_{a_i} = 1-\angle_{h,h_i},\quad \fpr_{a_i} = \angle_{h,h_i}.
	\end{equation}
	Therefore, for any investment levels $(\pi_{a_1}, \pi_{a_2})$, the firm solves the following profit maximization problem:
	\begin{align*}
	h^* &= \argmax_{h \in S_{d-1} } ~r\pi_{a_1}( 1-\angle_{h,h_1}) - (1-\pi_{a_1})\angle_{h,h_1} + r\pi_{a_2}( 1-\angle_{h,h_2}) - (1-\pi_{a_2})\angle_{h,h_2} \\
	&= \argmax_{h \in S_{d-1} } ~(1-r)\pi_{a_1}\angle_{h,h_1} + (1-r)\pi_{a_2}\angle_{h,h_2} - (\angle_{h,h_1}+\angle_{h,h_2}).
	\end{align*}
	
	The last term $\angle_{h,h_1}+\angle_{h,h_2}$ is minimized whenever $h$ is in the convex hull of $h_1$ and $h_2$. Then it is clear that for $r > 1$, the profit is maximized at $h = h_1$ whenever $\pi_{a_1} > \pi_{a_2}$, at $h=h_2$ whenever $\pi_{a_1} < \pi_{a_2}$, and at any $h$ in the convex hull of $h_1$ and $h_2$ whenever  $\pi_{a_1} = \pi_{a_2}$.
	
	To conclude that $h=h_1$ and $h=h_2$ are indeed stable equilibria, we check the best response \investmentlevels~by both groups at $h=h_1$ and $h=h_2$ satisfies the optimality conditions. For $h = h_1$, we have $\brpi_{a_1}(h_1) = G(w)$ and $ \brpi_{a_2}(h_1) = G\left(w \cdot (1-2\angle_{h_{1},h_2})\right)$, and indeed $\brpi_{a_1} >  \brpi_{a_2}$, by the monotonicity of $G$. For $h = h_2$, we have $\brpi_{a_1}(h_2) =G\left( w \cdot (1-2\angle_{h_{1},h_2})\right)$ and $ \brpi_{a_2}(h_2) = G(w) $, and indeed $\brpi_{a_1} <  \brpi_{a_2}$.
	
	Now we identify the unstable equilibrium, which is $h = \hm$, because the \investmentlevel~is indeed equal for both groups under this policy, i.e., we have \[ \brpi_{a_1} =  G(w(1-2\angle_{h_{1},\hm}) )=G(w(1-\angle_{h_{1},h_2}) )  = G(w(1-2\angle_{h_{2},\hm})) = \brpi_{a_2}. \] \lledit{When both groups are investing at this rate, we may assume that the institution's best response involves breaking ties among all utility-maximizing hyperplanes to choose $\hm$. This ensures that the dynamics are well-defined.} Notice that this is an unstable equilibrium, since this is the unique value of $h$ such that $\brpi_{a_2}(h) = \brpi_{a_2}(h)$, and any deviation from equal \investmentlevels~will change the profit-maximizing hyperplane to $h_1$ or $h_2$.
\end{proof}

\begin{proof}[Proof of Proposition \ref{prop:gaussian-limit} ]
	
	Following the proof of Proposition~\ref{prop:gaussian-eq}, we find that for $\ptp< \cfp$, the profit is maximized at $h=h_1$ whenever  $\pi_{a_1} <  \pi_{a_2}$, at $h=h_2$ whenever $\pi_{a_1} >  \pi_{a_2}$, and at any $h$ in the convex hull of $h_1$ and $h_2$ whenever  $\pi_{a_1} = \pi_{a_2}$.
	
	Checking the \investmentlevel~at $h=h_1$ (resp. $h=h_2$), we find that  $\brpi_{a_1}(h_1) = G(w)$ and $ \brpi_{a_2}(h_1) = G\left(w \cdot (1-\angle_{h_{1},h_2})\right)$, so  $\brpi_{a_1}(h_1) >  \brpi_{a_2}(h_1) $. Similarly, we have $\brpi_{a_1}(h_2)<  \brpi_{a_2}(h_2) $. This implies there is a 2-point limit cycle at $h=h_1$ and $h=h_2$.
	
	As before, the unstable equilibrium is at $h = \hm$, because the \investmentlevel~is indeed equal for both groups under this policy. Notice that this is an unstable equilibrium, since this is the unique value of $h$ such that $\brpi_{a_2}(h) = \brpi_{a_2}(h)$.
\end{proof}

\begin{table}[tbp]
	\begin{center}
		\begin{tabular}{|| p{2.5cm} |  p{3cm}|   p{3cm}|  p{3cm}| p{2.5cm}||} 
			\hline
			Equilibrium $h$ & $h_1$ & $h_2$ & $\hm$ & Ranking by metric \\ [0.5ex] 
			\hline\hline
			Stability  & Stable & Stable  &  Unstable & - \\ 
			\hline
			Qualification rate in group $a_1$, $\pi_{a_1}$  & $w$ &$w (1-2\angle)$ & $w(1-\angle)$ & $h_1 \succ \hm \succ h_2$  \\ 
			\hline
			Qualification rate in group $a_2$, $\pi_{a_2}$   & $w (1-2\angle)$  & $ w$ & $w(1-\angle)$  & $h_2 \succ \hm \succ h_1$   \\ 
			\hline
			Balance in qualification rate $|\pi_{a_1}-\pi_{a_2}|$ & $2w\angle$ & $2w\angle$ & $w\angle$ & $\hm\succ h_1 \sim h_2$  \\ 
			\hline
			%			Total income received by group $a_1$ & $w^2$  & $w^2 (1-2\angle)^2 + w\angle $ & $w^2 (1-\angle)^2 + \frac{1}{2}w\angle $ & $h_1 \succ h_2$ for $w > \frac{1}{4},\angle < 1- \frac{1}{4w}$; $h_m \succ h_2$ for $w > \frac{1}{4},\angle < \frac{2}{3}- \frac{1}{6w}$; $h_1 \succ h_m$ for $w > \frac{1}{4},\angle <2 - \frac{1}{2w}$\\
			%			\hline
			%			Total income received by group $a_2$ & $w^2 (1-2\angle)^2 + w\angle $ & $w^2$ & $w^2 (1-\angle)^2 + \frac{1}{2} w\angle $ & See above  \\
			%			\hline
			Institution's utility & $\ptp w (2-3\angle)+ 2(\ptp-\cfp)w\angle^2$ & $\ptp w (2-3\angle)+ 2(\ptp-\cfp)w\angle^2$ &  $\ptp  w (2-3\angle) + (\ptp-\cfp)w\angle^2$& $h_1\sim h_2 \succ \hm$ \\ [.5ex] 
			\hline
		\end{tabular}
	\end{center}
	\caption{Comparison of equilibria for Multivariate Gaussian features}\label{app:tab:gaussian}
\end{table}

\paragraph{Comparison of equilibria} In Table \ref{app:tab:gaussian}, we compare the equilibria described in Proposition~\ref{prop:gaussian-eq} on several metrics. We use $\angle$ to denote $\angle_{h_1,h_2}$.

	% !TeX root = main.tex 
\section{Supplementary material for Section \ref{sec:subsidies}}
\label{app:proof-subsidies}
\begin{proof}[Proof of Proposition \ref{prop:unequal_costs_eq} ]
	First notice that by assumption, $h = h_1$ cannot be at equilibrium. It is easy to check that $G_1(w(1-2\angle_{h_{1},h_2})) \le G_1(w) < G_2(w(1-2\angle_{h_{1},h_2})) \leq G_2(w)$, so $h=h_2$ is still at a stable equilibrium. Now, for any $h$ that is a convex combination of $h_1$ and $h_2$, we have $G_1(w(1-2\angle_{h,h_1})) \leq G_1(w) < G_2(w(1-2\angle_{h_{1},h_2})) \le G_2(w(1-2\angle_{h,h_2}))$, implying that $G_1(w(1-2\angle_{h,h_1})) \ne G_2(w(1-2\angle_{h,h_2}))$ for all $h$ that maximize institutional utility, so no other fixed points exist.
\end{proof}

	% !TeX root = appendix_supp.tex 
\section{Supplementary material and proofs for Section \ref{sec:non-realizable}
}\label{app:non-realizable}

The following result from \citet{coate93will} establishes conditions under which multiple equilibria exists for a single group when the features $\calX = [0,1]$ represent a score and the assessment rule is a threshold function. For completeness, we show that it can be derived as a consequence of Proposition~\ref{prop:multi_eq}.

%\todo{This is about a single $a$ so we should suppress $a$  rather than have it in $\Pr$ and $\pi$ but not in $F_a$ and $f_1$. Check to make sure these changes do not affect anything else.}
\begin{proposition}[Proposition 1 of \citet{coate93will}]\label{prop:coate_loury}
	\nhedit{Consider the case where $\calX = [0,1]$ is a space of one-dimensional scores, $\Theta=[0,1]$, and $\hat Y_\theta = \Ind{X > \theta}$ for all $\theta\in \Theta$. }
	Denote the conditional score CDFs as
	\nhedit{\[ F_1(x) \coloneqq \Pr(X<x \mid Y=1), ~ F_0(x) \coloneqq \Pr(X<x \mid Y=0). \]} 
	Let $f_1(x), f_0(x)$ be the \nhedit{point densities of $F_1$ and $F_0$, respectively.} Let $\phi(x) \coloneqq \frac{f_0(x)}{f_1(x)}$ be the likelihood ratio at $x$. Let $r \coloneqq \frac{\ptp}{\cfp}$ be the ratio of net gain to loss for the firm.
	Assume $\phi(x)$ is strictly decreasing --- \nhedit{i.e. as score increases, candidate is more likely to be skilled --- continuous and strictly positive on $[0,1]$.} %(i.e. we need $f_0, f_1 > 0$ on $[0,1]$)
	Further assume that $G(c)$  is continuous and $ G(w(F_0(\theta) - F_1(\theta))) > \frac{\phi(\theta)}{r + \phi(\theta)}$ for some $\theta \in (0,1)$. Then there exists at least two distinct non-zero equilibria.
	%\nhedit{Also, equilibrium $\pi$ with its corresponding assessment rule is locally stable whenever $\left|\frac{\partial \brtheta(\pi)}{\partial \pi}\right| < \left|\frac{\partial \brpi(\theta)}{\partial \theta}\right| $.}
\end{proposition}

\begin{proof}[Proof of Proposition~\ref{prop:coate_loury}]
Note that for any $\theta$, $\tpr(\theta) - \fpr(\theta) = F_0(\theta) - F_1(\theta)$.  Therefore, the group's qualification rate in response to assessment parameter $\theta$ is 
		\begin{equation}
		\brpi(\theta) = G(w(F_0(\theta) - F_1(\theta))).
		\end{equation}
	
	Since $\phi(x)$ is strictly decreasing, the utility maximizing assessment rule $\theta$ in response to the qualification rate $\pi$ is
	\begin{equation}
	\brtheta(\pi) = \inf\left\{x\in[0,1]: r \ge \frac{1-\pi}{\pi}\cdot \phi(x) \right\}.
	\end{equation}
	Since $\phi(x)$ is continuous and strictly positive, we must also have that $F_0, F_1$ are continuous, and so in particular $\beta(\pi) = F_0(\brtheta(\pi)) - F_1(\brtheta(\pi)) $ is continuous. By assumption, $G$ is continuous and there exists  $x \in (0,1)$  such that $x< G(w\beta(x))$ since $\brtheta(\pi)$ is surjective. Therefore, the claim follows from Proposition~\ref{prop:multi_eq}.
\end{proof}
	% !TeX root = main.tex 
\section{Supplementary material for Section \ref{sec:experiments}
 }\label{app:expt}

We collect here additional figures for Section \ref{sec:experiments}. 

\begin{figure}[htbp]
	\centering
	\includegraphics[width=.7\columnwidth]{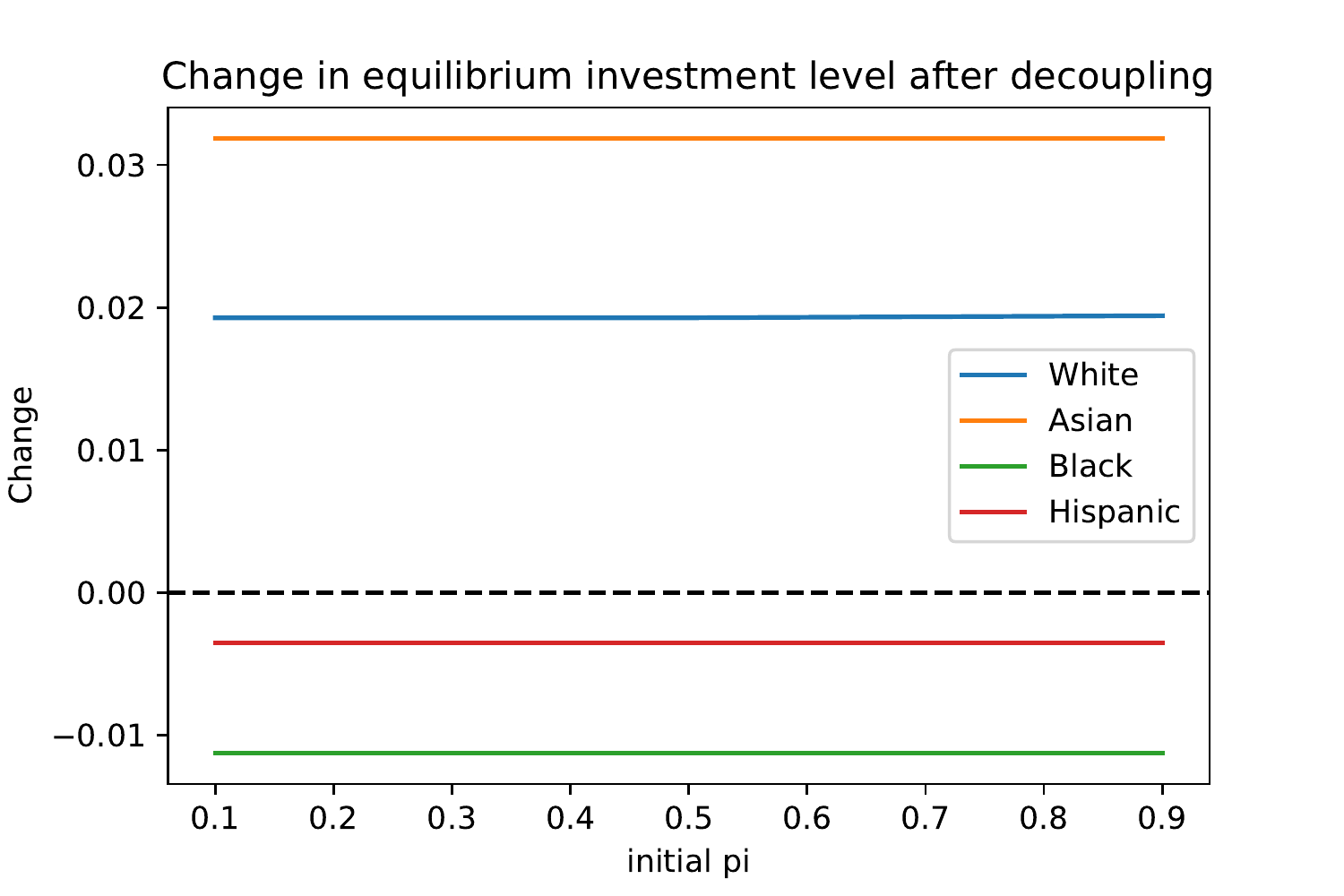}%
	\caption{Effects of decoupling without multiple equilibria. $G$ is the uniform distribution on $[0,1]$ for all groups and the reward is $w=1$. The decoupled equilibria are unique for this choice of $G$.}\label{fig:decoupling1}
\end{figure}

\begin{figure}[htbp]
	\centering
	\includegraphics[width=.5\columnwidth]{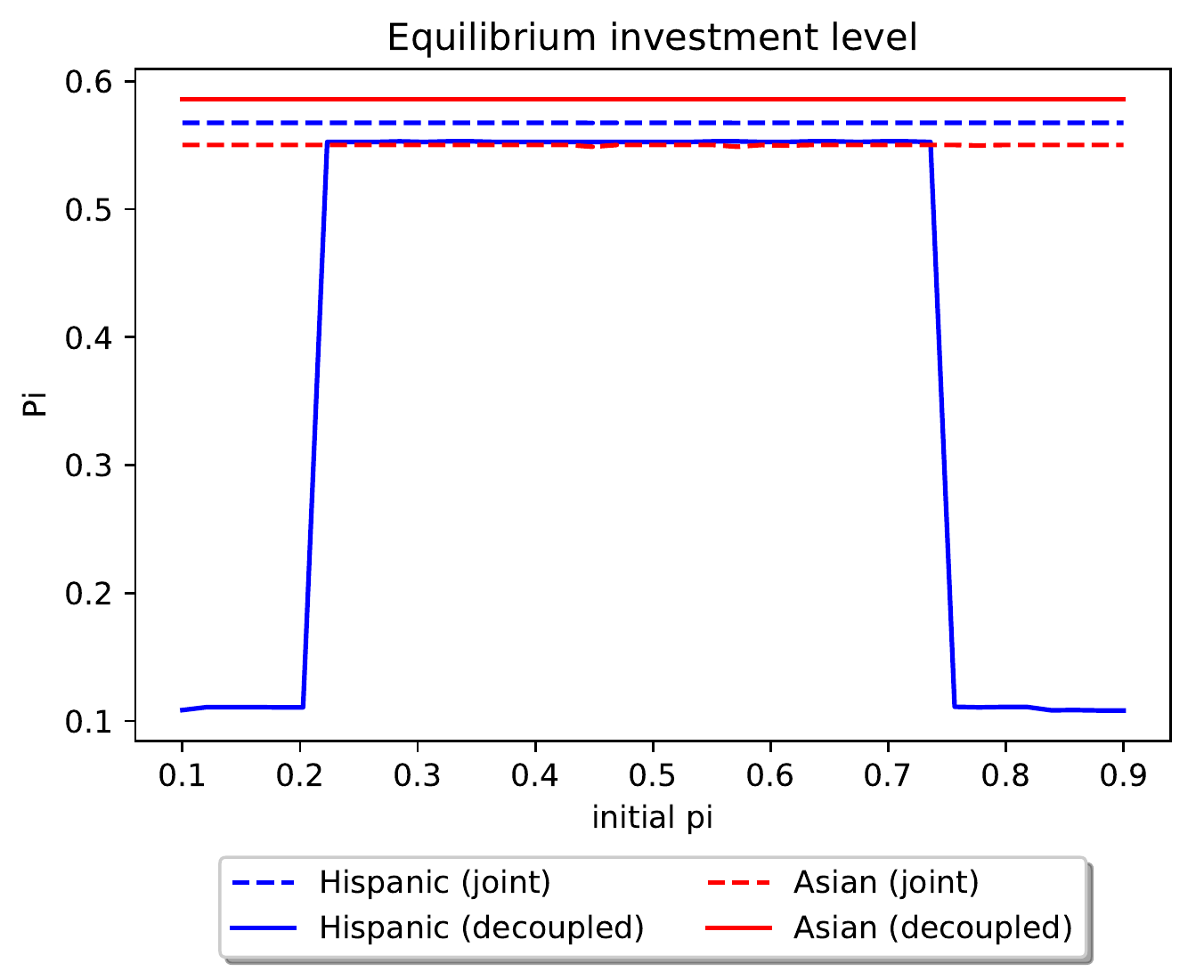}%
	\includegraphics[width=.5\columnwidth]{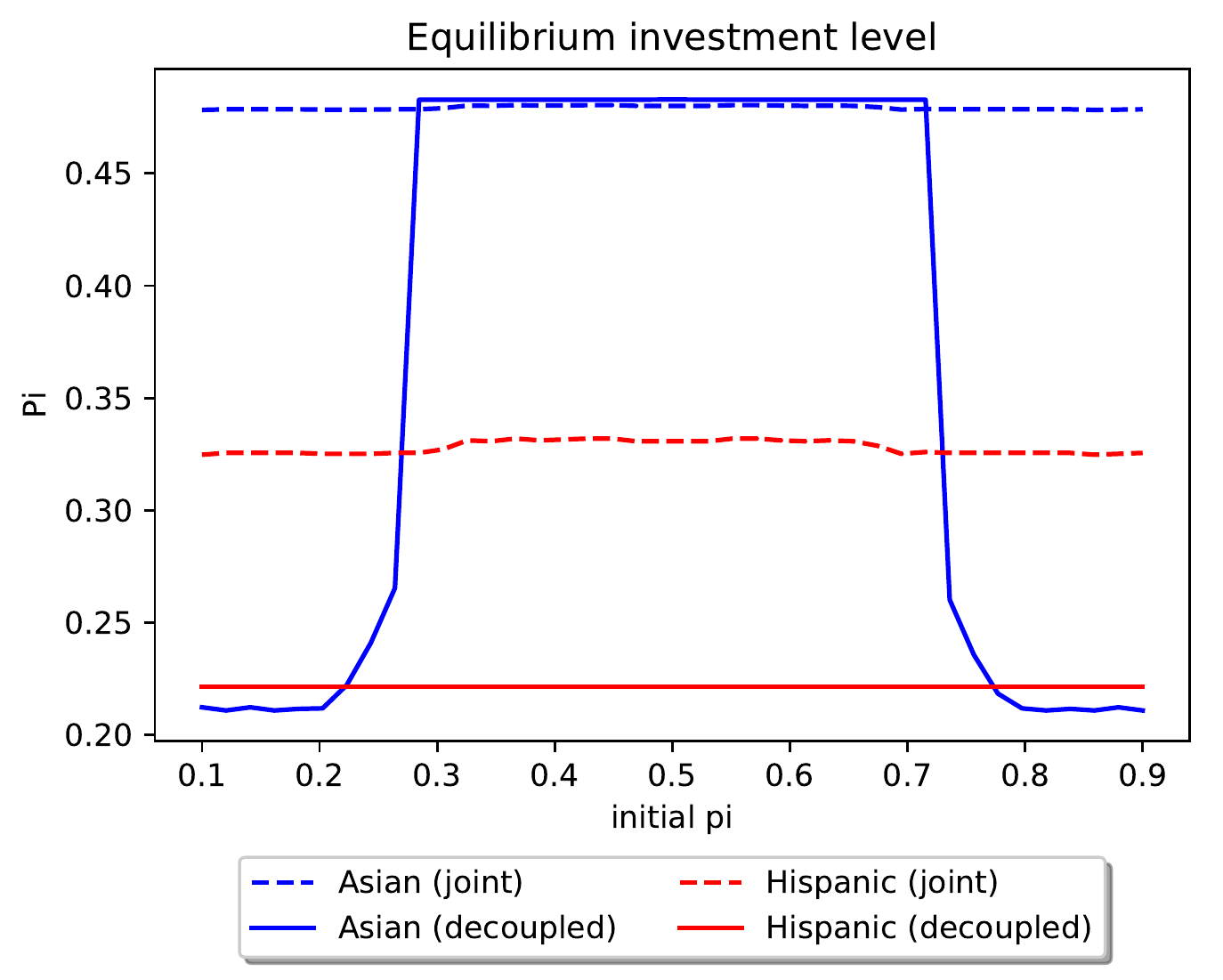}
	\caption{Effects of decoupling in presence of multiple equilibria. We vary the initial level of investment in the x-axis. A different bimodal Gaussian distribution $G$ was used to generated each plot.}\label{fig:decoupling-opt1}
\end{figure}

\end{document}